\documentclass[11pt]{article}

\usepackage{fullpage}
\usepackage{amsmath}
\usepackage{amsthm}
\usepackage{amsfonts}
\usepackage{amssymb}
\usepackage{mathrsfs}
\usepackage{hyperref}
\usepackage{xcolor}

\newtheorem{lemma}{Lemma}

\newtheorem{corollary}{Corollary}
\newtheorem{theorem}{Theorem}
\newtheorem{proposition}{Proposition}

\newcounter{example}
\newenvironment{example}{\refstepcounter{example}\par\bigskip
\noindent\textit{Example~\theexample.} \rmfamily}{\hfill$\dashv$\bigskip}

\newcommand{\supp}{{\mathrm{Supp}}}
\newcommand{\tuples}{{\mathrm{Tup}}}
\newcommand{\domain}{{\mathrm{Dom}}}
\newcommand{\Entropy}{{\mathrm{H}}}
\newcommand{\KL}{{\mathrm{D}}}
\newcommand{\ltgc}{local-to-global consistency property}

\newcommand{\commentout}[1]{}
\usepackage{authblk}

\begin{document}

\title{{\bf Consistency, Acyclicity, and Positive Semirings}}
\author[1]{Albert Atserias}
\author[2] {Phokion G.\ Kolaitis}
\affil[1]{Universitat Polit\`{e}cnica de Catalunya}
\affil [2]{University of California Santa Cruz and IBM Research}
\date{September 20, 2020}

\maketitle

\begin{abstract}
In several different settings, one comes across situations in which
the objects of study are locally consistent but globally
inconsistent. Earlier work about probability distributions by Vorob'ev
(1962) and about database relations by Beeri, Fagin, Maier, Yannakakis
(1983) produced characterizations of when local consistency always
implies global consistency. Towards a common generalization of these
results, we consider K-relations, that is, relations over a set of
attributes such that each tuple in the relation is associated with an
element from an arbitrary, but fixed, positive semiring K. We
introduce the notions of projection of a K-relation, consistency of
two K-relations, and global consistency of a collection of
K-relations; these notions are natural extensions of the corresponding
notions about probability distributions and database relations. We
then show that a collection of sets of attributes has the property
that every pairwise consistent collection of K-relations over those
attributes is globally consistent if and only if the sets of
attributes form an acyclic hypergraph. This generalizes the
aforementioned results by Vorob'ev and by Beeri et al., and
demonstrates that K-relations over positive semirings constitute a
natural framework for the study of the interplay between local and
global consistency. In the course of the proof, we introduce a notion
of join of two K-relations and argue that it is the ``right''
generalization of the join of two database relations. Furthermore, to
show that non-acyclic hypergraphs yield pairwise consistent
K-relations that are globally inconsistent, we generalize a
construction by Tseitin (1968) in his study of hard-to-prove
tautologies in propositional logic.
\end{abstract}

\section{Introduction}

There are many situations, spanning art and science, in which the
objects under consideration are locally consistent but globally
inconsistent, where the terms ``local", ``global", and ``consistent"
are used in some intuitive sense but can be made precise in each
concrete setting.  In art, Escher's~1960 \emph{Ascending and
  Descending} and~1961 \emph{Waterfall} lithographs are striking
depictions of locally consistent but globally inconsistent
situations. Closely related to Escher's artwork is the work by
L.S.~Penrose and R.~Penrose~\cite{penrose1958impossible} on impossible
objects, such as the impossible tribar (see
also~\cite{francis2007impossible}). In quantum mechanics, the
interplay between local consistency and global inconsistency takes the
form of non-locality and contextuality phenomena, where collections of
empirical local measurements may not admit a global explanation via a
hidden variable; prominent results in this area include Bell's
Theorem~\cite{bell1964einstein} and Hardy's
paradox~\cite{hardy1992quantum}. In probability theory, there is work
on when a given collection of pairwise consistent probability
distributions admits a global distribution whose marginal
distributions coincide with the given
collection~\cite{vorob1962consistent}. In computer science, the
interplay between local consistency and global consistency arises in
such different areas as constraint
satisfaction~\cite{DBLP:books/daglib/0016622}, proof
complexity~\cite{ChvatalSzemeredi1988}, and relational
databases~\cite{BeeriFaginMaierYannakakis1983}.

What do the aforementioned situations have in common and is there a
unifying framework behind them?
Abramsky~\cite{DBLP:conf/birthday/Abramsky13,DBLP:journals/eatcs/Abramsky14}
pointed out that there are formal connections between non-locality and
contextuality in quantum mechanics on one side and the universal
relation problem in database theory on the other side.  The latter is
the following decision problem: given a collection~$X_1,\ldots,X_m$ of
sets of attributes (that is, names of columns of relations) and a
collection~$R_1,\ldots,R_m$ of relations over~$X_1,\ldots,X_m$ (that
is,~$X_i$ is the set of the attributes of~$R_i$, for~$i\in [m]$), are
the relations~$R_1,\ldots,R_m$ globally consistent? In other words, is
there a relation~$R$, called a \emph{universal} relation,
over~$X_1\cup \cdots \cup X_m$ such that, for every~$i \in [m]$, the
projection~$R[X_i]$ of~$R$ on~$X_i$ is equal to~$R_i$?  Clearly, if
such a universal relation exists, then the relations~$R_1,\ldots,R_m$
are \emph{pairwise consistent},
i.e.,~$R_i[X_i\cap X_j] = R_j[X_i\cap X_j]$, for all~$i,j \in [m]$,
but the converse need not hold.  Switching to the quantum mechanics
side and by regarding the collection of empirical measurements
in~\cite{hardy1992quantum} as a collection of database relations,
Hardy's paradox can be viewed as a negative instance of the universal
relation problem: the database relations at hand are pairwise
consistent, but globally inconsistent.  Note that, since experiments
are typically repeated, measurements give rise to probabilities. This
way, Bell's Theorem~\cite{bell1964einstein} can be viewed as an
instance of a collection of probability distributions that are
pairwise consistent, but globally inconsistent.  As regards unifying
frameworks, Abramsky and
Brandenburger~\cite{DBLP:journals/corr/abs-1102-0264} used sheaf
theory to provide a unified account of non-locality and
contextuality. This approach was explored further
in~\cite{DBLP:journals/corr/abs-1111-3620,DBLP:conf/csl/AbramskyBKLM15}.

As mentioned in the preceding paragraph, pairwise consistency is a
necessary, but not sufficient, condition for a collection of relations
to be globally consistent.  In the setting of relational databases,
Beeri, Fagin, Maier, and
Yannakakis~\cite{BeeriFaginMaierYannakakis1983} characterized when
pairwise consistency is also a sufficient condition for global
consistency.  If~$X_1,\ldots,X_m$ are sets of attributes, we say that
the collection~$X_1,\ldots,X_m$ has the \emph{\ltgc}~if every
collection~$R_1,\ldots,R_m$ of pairwise consistent relations
over~$X_1,\ldots,X_m$ is globally consistent. The main finding in
Beeri et al.~\cite{BeeriFaginMaierYannakakis1983} is that a
collection~$X_1,\ldots,X_m$ of sets of attributes has the \ltgc~if and
only if the hypergraph with~$X_1,\ldots,X_m$ as hyperedges is acyclic,
where the notion of hypergraph acyclicity is a suitable generalization
of the notion of graph acyclicity. Observe that the \ltgc~is a
semantic property (in the sense that its definition involves relations
over the sets of attributes), while acyclicity is a syntactic property
(in the sense that it describes a structural property of hypergraphs
with no reference to
relations). In~\cite{BeeriFaginMaierYannakakis1983}, several other
syntactic conditions on hypergraphs were considered, and each was
shown to be equivalent to acyclicity. In the setting of probability
theory, Vorob'ev~\cite{vorob1962consistent} identified a different
syntactic condition on hypergraphs, which we call \emph{Vorob'ev
  regularity}, and showed that a collection of probability
distributions over~$X_1,\ldots,X_m$ has the \ltgc~(suitably adapted to
probability distributions) if and only the hypergraph
with~$X_1,\ldots,X_m$ as hyperedges is Vorob'ev regular. It is perhaps
worth noting that Vorob'ev's paper~\cite{vorob1962consistent} was
published much earlier, but Beeri et
al.~\cite{BeeriFaginMaierYannakakis1983} were apparently unaware of
Vorob'ev's work.  It is now natural to ask: is there a common
generalization of the above results? This question was investigated by
Barbosa in his doctoral thesis~\cite[Chapter
VI]{Barbosa-thesis}. Barbosa explored the question in the
sheaf-theoretic framework for non-locality and contextuality and showed
that hypergraph acyclicity implies the \ltgc~in that framework, but
did not obtain the reverse direction.

We establish a common generalization of the results by
Vorob'ev~\cite{vorob1962consistent} and by Beeri et
al.~\cite{BeeriFaginMaierYannakakis1983}. Instead of the
sheaf-theoretic framework, we work in the algebraic framework of
\emph{positive semirings}, which are commutative semirings with no
zero-divisors and with the property that~$a+b=0$ holds for two
elements~$a$ and~$b$ of the semiring if and only if~$a=b=0$. Positive
semirings were used to study the provenance of relational database
queries~\cite{DBLP:conf/pods/GreenKT07} and also the provenance of
first-order sentences~\cite{DBLP:journals/corr/abs-1712-01980};
furthermore, commutative semirings were considered by
Abramsky~\cite{DBLP:conf/birthday/Abramsky13} in discussing algebraic
databases as a generalization of relational databases.

Let~$K$ be a positive semiring. As a common generalization of database
relations and probability distributions, we
consider~$K$-\emph{relations}, i.e., relations over a set of
attributes such that each tuple in the relation has an associated
element from~$K$ as value. Note that ordinary relations
are~$K$-relations where~$K$ is the Boolean semiring, while probability
distributions are~$K$-relations with~$K$-values adding to~$1$ and
where~$K$ is the semiring of the non-negative real numbers. We
introduce natural extensions of the notions of projection of
a~$K$-relation, pairwise consistency, global consistency, and the
\ltgc~for~$K$-relations.  We then show that a
collection~$X_1,\ldots,X_m$ of sets of attributes has the
\ltgc~for~$K$-relations if and only if the hypergraph
with~$X_1,\ldots,X_m$ as hyperedges is acyclic. We also show that a
hypergraph is Vorob'ev regular if and only if it is acyclic (this
result has been mentioned in passing or has been taken for granted in
earlier papers, but we have not found an explicit reference for
it). The results by Vorob'ev~\cite{vorob1962consistent} and by Beeri
et al.~\cite{BeeriFaginMaierYannakakis1983} then follow as immediate
corollaries.

While the proof of our main result about the equivalence between
hypergraph acyclicity and the \ltgc~for~$K$-relations bears some
similarities and analogies with the earlier proofs of its special
cases, it also brings in some new concepts and tools that may be of
independent interest. We conclude this section by highlighting some of
these concepts and tools.

To prove that hypergraph acyclicity implies the local-to-global
consistency property for~$K$-relations, we introduce a \emph{join}
operation on~$K$-relations. We make the case that this is the
``right'' extension to~$K$-relations of the notion of the join of two
ordinary relations. In particular, we show that the join of two
consistent~$K$-relations witnesses their consistency and also that the
basic results about lossless-join decompositions of ordinary relations
extend to~$K$-relations. Note that if~$K$ is the semiring of
non-negative integers, then the~$K$-relations are precisely the
\emph{bags} (also known as \emph{multisets}). Our join operation on
bags is, in general, different from the standard bag join used in SQL
(for bag operations in SQL, see~\cite{DBLP:books/daglib/0011318}). We
point out, however, that unlike the join operation introduced here,
the standard bag join does not always witness the consistency of two
consistent bags. Furthermore, we show that the join of two consistent
probability distributions is the unique probability distribution that
maximizes entropy among all probability distributions that witness the
consistency of the two probability distributions we started with.

To prove that the \ltgc~for~$K$-relations implies hypergraph
acyclicity, we need to have a systematic way to produce negative
instances of the universal relation problem, such as the instances
found in Hardy's paradox and related constructions in the study of
non-locality and contextuality. Note that, in our setting, we need the
relations in the negative instances to be~$K$-relations where~$K$ is
an \emph{arbitrary} positive semiring, instead of ordinary relations
over the Boolean semiring or probability distributions over the
semiring of nonnegative real numbers; furthermore, we need to be able
to produce such negative instance for \emph{any} given cyclic
collection~$X_1,\ldots,X_m$ of sets of attributes.  For the special
cases of ordinary relations and probability distributions, Beeri et
al.~\cite{BeeriFaginMaierYannakakis1983} and
Vorob'ev~\cite{vorob1962consistent} provided suitable such
constructions, which, as far as we can tell, do not generalize to
arbitrary positive semirings. For our construction, which works for an
arbitrary positive semiring, we adapt an idea that can be traced to
Tseitin~\cite{Tseitin1968} in his study of hard-to-prove tautologies
in propositional logic. In brief, Tseitin constructed arbitrarily
large sets of propositional clauses such that any fixed number of them
are satisfiable, but, when taken jointly, they are unsatisfiable.  The
combinatorial principle underlying Tseitin's construction is the
following basic~\emph{parity principle}: for every undirected graph
and for every labeling of the vertices of the graph with 0's and 1's
with an odd total number of 1's, there is no subset of the edges that
touches every vertex a number of times that is congruent to the label
of the vertex modulo~$2$. To generalize this to arbitrary cyclic
hypergraphs and to arbitrary semirings, we resort to a similar modular
counting principle for a modulus~$d \geq 2$ that depends on the
structure of the hyperedges~$X_1,\ldots,X_m$. While similar but
different variations of Tseitin's construction have been used in other
contexts (see, e.g.,~\cite{DBLP:journals/jcss/BussGIP01}
and~\cite{DBLP:journals/tcs/AtseriasBD09}), we are not aware of any
other construction that simultaneously generalizes the results in
Beeri et al.~\cite{BeeriFaginMaierYannakakis1983} and
Vorob'ev~\cite{vorob1962consistent}.  Furthermore, it is worth noting
that our construction contains as a special case the most
basic~Popescu-Rorhlich~box~\cite{popescu1994quantum}, which is another
well-known example of non-locality and contextuality (see,
e.g.,~\cite{DBLP:journals/corr/abs-1102-0264}). Specifically, the
support of the Popescu-Rorhlich~box is precisely the special case of
our construction in which the hypergraph~$X_1,\ldots,X_m$ is the
4-cycle~$AB,BC,CD,DA$ on the four vertices~$A,B,C,D$.

\section{Valued Relations up to Normalization} \label{sec:valuedrelations}

In this section we define the notion of \emph{valued relation},
or~\emph{$K$-relation} for a positive semiring~$K$ of values, as a
generalization of the database-theoretic notion of relation. We study
its most basic properties and discuss some examples. Besides the
standard concept of ordinary relation from database theory, two other
canonical examples will be the \emph{bags} and the
\emph{probability distributions}.

\subsection{Definition of Valued Relations and Their Basic Properties}

We start by recalling some basic terminology and notation from the
theory of databases. While most of our notation is standard and
well-established, we refer to the standard
textbooks~\cite{DBLP:books/cs/Ullman88}
and~\cite{DBLP:books/aw/AbiteboulHV95} for further elaboration.

\paragraph{Attributes, Tuples, and Relations}
An \emph{attribute}~$A$ is a symbol with an associated
set~$\domain(A)$ called its \emph{domain}. If~$X$ is a finite set of
attributes, then we write~$\tuples(X)$ for the set
of~\emph{$X$-tuples}; i.e.,~$\tuples(X)$ is the set of maps that take
each attribute~$A \in X$ to an element of its
domain~$\domain(A)$. Note that~$\tuples(\emptyset)$ is non-empty as it
contains the \emph{empty tuple}, i.e., the unique map with empty
domain. If~$Y \subseteq X$ is a subset of attributes and~$t$ is
an~$X$-tuple, then the \emph{projection of~$t$ on~$Y$}, denoted
by~$t[Y]$, is the unique~$Y$-tuple that agrees with~$t$ on~$Y$. In
particular,~$t[\emptyset]$ is the empty tuple.

 A \emph{relation
  over~$X$} is a subset of~$\tuples(X)$; it is a finite relation if it
is a finite subset of~$\tuples(X)$. In what follows, we will often refer to such relations
as \emph{ordinary} relations to differentiate them from $K$-relations, where $K$ is a positive semiring other than the Boolean semiring.
We write~$R(X)$ to emphasize the
fact that the relation~$R$ has \emph{schema}~$X$. In this paper all
sets of attributes and all relations are finite, so we omit the term.
If~$Y \subseteq X$ and~$R$ is a relation over~$X$, then the
\emph{projection of~$R$ on~$Y$}, denoted $R[Y]$, is the relation over~$Y$ made of all
the projections~$t[Y]$ as~$t$ ranges over~$R$.
 If~$R$ is a relation over~$X$
and~$S$ is a relation over~$Y$, then their \emph{join}~$R \Join S$ is
the relation over~$X \cup Y$ made of all the~$X \cup Y$-tuples~$t$
such that~$t[X]$ is in~$R$ and~$t[Y]$ is in~$S$.

If~$X$ and~$Y$ are sets of attributes, then we write~$X\!Y$ as
shorthand for the union~$X \cup Y$. Accordingly, if~$x$ is
an~$X$-tuple and~$y$ is a~$Y$-tuple with the property
that~$x[X \cap Y] = y[X \cap Y]$, then we write~$xy$ to denote
the~$X\!Y$-tuple that agrees with~$x$ on~$X$ and on~$y$ on~$Y$.  We
say that~\emph{$x$ joins with~$y$}, and that~\emph{$y$ joins
  with~$x$}, to \emph{produce} the tuple~$xy$.

\paragraph{Positive Semirings} A \emph{commutative semiring} is a
set~$K$ with two binary operations~$+$ and~$\times$ that are
commutative, associative, have $0$ and $1$, respectively, as identity elements,~$\times$ distributes over~$+$, and~$0$
\emph{annihilates}~$K$, that is,~$0 \times a = a \times 0 = 0$ holds for
all~$a \in K$. We assume that~$0 \not= 1$, that is, the semiring is
\emph{non-trivial}. The identity of multiplication~$1$ is also called
the \emph{unit} of the semiring. We write multiplication~$a \times b$
by concatenation~$ab$ or with a dot~$a \cdot b$. If there do not exist
non-zero~$a$ and~$b$ in~$K$ such that~$a+b=0$, then we say that~$K$ is
\emph{plus-positive}. If there do not exist non-zero~$a$ and~$b$
in~$K$ such that~$ab = 0$, then we say that~\emph{$K$ has no
  zero-divisors}, or that~$K$ is a semiring \emph{without
  zero-divisors}. A plus-positive commutative semiring without
zero-divisors is called \emph{positive}. In the sequel,~$K$ will always
denote a non-trivial positive commutative semiring.

We introduce some examples.  The \emph{Boolean
  semiring}~$\mathbb{B} = (\{0,1\},\vee,\wedge,0,1)$ has $0$ (false)
and $1$ (true) as elements, and disjunction ($\vee$) and conjunction
($\wedge$) as operations. This is a commutative semiring that is
plus-positive and has no zero-divisors, hence it is  positive; it is not a
ring since disjunction does not have an inverse. The non-negative
integers~$\mathbb{Z}^{\geq 0}$, the non-negative
rationals~$\mathbb{Q}^{\geq 0}$, and the non-negative
reals~$\mathbb{R}^{\geq 0}$ with their usual arithmetic operations~$+$
and~$\times$ and their identity elements~$0$ and~$1$ are also positive
semirings. In contrast, the full integers~$\mathbb{Z}$, the
rationals~$\mathbb{Q}$, or the reals~$\mathbb{R}$ are commutative
semirings without zero-divisors that are not plus-positive. The
semiring of non-negative integers~$\mathbb{Z}^{\geq 0}$ is also
denoted by~$\mathbb{N}$, and it is called the \emph{bag semiring}.
For an integer~$m \geq 2$, the semiring~$\mathbb{Z}_m$ of arithmetic
mod~$m$, also denoted by~$\mathbb{Z}/m\mathbb{Z}$, is a commutative
semiring that is not plus-positive, and that has no zero-divisors if
and only if~$m$ is prime;~$\mathbb{Z}_1$ is not even non-trivial.

Under the convention that~$0 < 1$, the disjunction~$\vee$ and
conjunction~$\wedge$ operations of the Boolean semiring can also be
written as~$\max$ and~$\min$, respectively. Semirings over arithmetic
ordered domains that combine the~$\max$ or~$\min$ operations with the
usual arithmetic operations are called \emph{tropical} semirings. The
\emph{min-plus} semiring has the extended
reals~$\mathbb{R} \cup \{ +\infty,-\infty \}$ as elements, and the
standard operations of minimum and addition for~$+$ and~$\times$,
with~$-\infty$ playing the role of the identity for~$\min$. The
\emph{positive min-plus} semiring  has the extended positive
reals~$\mathbb{R}^{\geq 0} \cup \{+\infty\}$ as elements, and again
standard minimum and addition for~$+$ and~$\times$, with~$0$ playing
the role of identity for~$\min$.  The \emph{Viterbi} semiring elements ranging over the
unit interval~$[0,1]$, and the standard operations of
maximum and multiplication for~$+$ and~$\times$, respectively. The
\emph{rational tropical} semirings are based on the extended rational numbers
in place of the extended real numbers.

\paragraph{Definition of~$K$-relations and Their Marginals}
Let~$K = (K^*,+,\times,0,1)$ be a semiring and let~$X$ be a finite set
of attributes. A~\emph{$K$-relation over~$X$} is a
map~$R : \tuples(X) \rightarrow K$ that assigns a value~$R(t)$ in~$K$
to every~$X$-tuple~$t$ in~$\tuples(X)$. Note that this definition
makes sense even if~$X$ is the empty set of attributes; in such a
case, a~$K$-relation over~$X$ is simply a single value from~$K$ that
is assigned to the empty tuple. Note also that the ordinary relation are precisely the $\mathbb B$-relation, where $\mathbb B$~is the Boolean semiring.

The \emph{support} of
the~$K$-relation~$R$, denoted by~$\supp(R)$, is the set
of~$X$-tuples~$t$ that are assigned non-zero value, i.e.,
\begin{equation}
\supp(R) := \{ t \in \tuples(X) : R(t) \not= 0 \}. \label{def:support}
\end{equation}
Whenever this does not lead to confusion, we write~$R'$ to
denote~$\supp(R)$. Note that~$R'$ is an ordinary relation
over~$X$. A~$K$-relation is \emph{finitely supported} if its support is a
finite set. In this paper, all~$K$-relations are finitely supported
and we omit the term. When~$R'$ is empty we say that~$R$ is the
empty~$K$-relation over~$X$.  For~$a \in K$, we write~$aR$ to denote
the~$K$-relation over~$X$ defined by~$(aR)(t) = aR(t)$ for
every~$X$-tuple~$t$. It is always the case
that~$\supp(aR) \subseteq \supp(R)$, and, as the proof of the next lemma shows, the
reverse inclusion~$\supp(R) \subseteq \supp(aR)$ also holds in
case~$a$ is a non-zero element of $K$ and~$K$ has no zero-divisors.
If~$t$ is a~$Y$-tuple for some~$Y \subseteq X$, then the
\emph{marginal of~$R$ over~$t$} is defined by
\begin{equation}
R(t) := \sum_{r \in R': \atop r[Y] = t} R(r). \label{eqn:marginal}
\end{equation}
Accordingly, any~$K$-relation over~$X$ induces a~$K$-relation over~$Y$
for any~$Y \subseteq X$. This~$K$-relation is denoted
by~$R[Y]$  and is called the \emph{marginal of~$R$
  on~$Y$}. Note that if $R$ is an ordinary relation (i.e., $R$ is a $\mathbb B$-relation), then the marginal $R[Y]$ is the projection of $R$ on $Y$, so the notation for the marginal is consistent with the one introduced for the projection earlier.
  It is always the case
that~$\supp(R[Y]) \subseteq \supp(R)[Y]$, and, as the proof of the next lemma shows, the
reverse inclusion also holds in case the semiring~$K$ is plus-positive.

From now on, we make the blanket assumption that $K$ is a positive semiring.  This hypothesis will not be explicitly spelled out in the statements of the various lemmas in which $K$-relations are mentioned.

\begin{lemma} \label{lem:easyfacts1} Let $R(X)$ be a~$K$-relation. The
  following statements hold:
  \begin{enumerate} \itemsep=0pt
  \item For all non-zero elements $a$ in $K$, we have $(aR)' = R'$.
  \item For all~$Y \subseteq X$, we have~$R'[Y] = R[Y]'$.
  \item For all $Z \subseteq Y \subseteq X$, we have $R[Y][Z] = R[Z]$.
\end{enumerate}
\end{lemma}

\begin{proof}
  For~1, the inclusion~$(aR)' \subseteq R'$ holds for all semirings
  since if~$t \in (aR)'$, then~$aR(t) \not= 0$, so~$R(t) \not= 0$
  since~$0$ annihilates~$K$, and hence~$t \in R'$. For the converse,
  if~$t \in R'$, then~$R(t) \not = 0$, so~$aR(t) \not= 0$ since~$a$ is
  non-zero and~$K$ has no zero-divisors, and hence~$t \in (aR)'$.
  For~2, the inclusion~$R[Y]' \subseteq R'[Y]$ is obvious and holds
  for all semirings. For the converse, assume that~$t \in R'[Y]$, so
  there exists~$r$ such that~$R(r) \not= 0$ and~$r[Y] =
  t$. By~\eqref{eqn:marginal} and the plus-positivity of~$K$ we have
  that~$R(t) \not= 0$. Hence~$t \in R[Y]'$.  For~3, we have
  \begin{equation}
  R[Y][Z](u) = \sum_{v \in R[Y]': \atop v[Z]=u} R[Y](v) =
  \sum_{v \in R'[Y]: \atop v[Z]=u} \sum_{w \in R': \atop w[Y]=v} R(w) =
  \sum_{w \in R': \atop w[Z]=u} R(w) = R[Z](u)
\end{equation}
where the first equality follows from~\eqref{eqn:marginal}, the second
follows from Part~2 of this lemma to replace~$R[Y]'$ by~$R'[Y]$, and
again~\eqref{eqn:marginal}, the third follows from partitioning the
tuples in~$R'$ by their projection on~$Y$, together
with~$Z \subseteq Y$, and the fourth follows from~\eqref{eqn:marginal}
again.
\end{proof}

Some examples follow.

\begin{example} \label{ex:examplessemirings}
When~$K$ is the Boolean semiring~$\mathbb{B}$,
a~$\mathbb{B}$-relation over~$X$ is simply an ordinary relation
over~$X$; its support is the relation itself, and its marginals are
the ordinary projections. When~$K$ is the bag semiring~$\mathbb{N}$,
the~$\mathbb{N}$-relations are called \emph{bags} or
\emph{multi-sets}. If~$T$ is a bag and~$t$ is a tuple in its support,
then~$T(t)$ is called the \emph{multiplicity} of~$t$ in~$T$. When~$K$
is the semiring of non-negative reals~$\mathbb{R}^{\geq 0}$, the
finite~$\mathbb{R}^{\geq 0}$-relations~$T$ that satisfy
\begin{equation}
  T[\emptyset] = \sum_{t \in T'} T(t) = 1 \label{eqn:normalizationequation}
\end{equation}
are the probability distributions of finite support over the
set~$\tuples(X)$ of~$X$-tuples, or, in short, the \emph{probability
  distributions over~$X$}. Conversely, to every finite
non-empty~$\mathbb{R}^{\geq 0}$-relation~$T$ one can associate a probability
distribution~$T^*$ through \emph{normalization}; this means that if we
set~$N_T := \sum_{t \in T'} T(t)$ and~$n_T = 1/N_T$, then the
$\mathbb{R}^{\geq 0}$-relation~$T^* := n_T T$ is a probability distribution. Finally, when~$K$ is the
semiring of non-negative rationals~$\mathbb{Q}^{\geq 0}$, the
corresponding probability distributions are called \emph{rational} probability distributions.
\end{example}

\subsection{Equivalence of~$K$-Relations}

We introduce a notion of equivalence between two~$K$-relations over
the same set of attributes that will play an important role in the
later sections of this paper. To motivate this definition,
let us look again at the probability distributions seen
as the~$\mathbb{R}^{\geq 0}$-relations that satisfy the normalization
equation~\eqref{eqn:normalizationequation} from
Example~\ref{ex:examplessemirings}.

\paragraph{Derivation of the Equivalence Relation}
Recall that to every non-empty~$\mathbb{R}^{\geq 0}$-relation~$T$ one
can associate a probability distribution~$T^*$ through
normalization~$T \mapsto T^*$. More generally, a normalization
operation can be defined for any semiring~$K$ that is actually a
\emph{semifield}, which is a semiring whose multiplication
operation~$\times$ admits an inverse~$\div$. Note
that both~$\mathbb{R}^{\geq 0}$ and~$\mathbb{Q}^{\geq 0}$ are
semifields. Formally, if~$K$ is a semifield and~$T$ is a
non-empty~$K$-relation over a set of attributes~$X$, then we
define~$T^*$ as the~$K$-relation defined by~$T^*(t) := (1/N_T) T(t)$
for every~$X$-tuple~$t$,
where~$N_T := T[\emptyset] = \sum_{t \in T'} T(t)$, and~$1/N_T$ is the
multiplicative inverse of~$N_T$ in the semifield~$K$. Note that~$T$
was assumed non-empty, so~$N_T \not= 0$ since~$K$ is plus-positive,
and the multiplicative inverse~$1/N_T$ exists. When~$T$ is the
empty~$K$-relation, we let~$T^*$ be the empty~$K$-relation itself.

With this definition in hand, still assuming that~$K$ is a semifield,
we can define an equivalence relation~$R \equiv S$ to hold between
two~$K$-relations~$R$ and~$S$ if and only if~$R^* = S^*$. An important
observation that follows from the definitions is that if~$R$ and~$S$
are~$K$-relations over the same set of attributes, then~$R^*=S^*$
holds if and only if~$aR = bS$ for some non-zero~$a$ and~$b$
in~$K$. For the \emph{only if} direction just take~$a = 1/N_R$
and~$b = 1/N_S$ if both~$R$ and~$S$ are non-empty~$K$-relations,
and~$a = b = 1$ otherwise. For the \emph{if} direction, assuming
that~$R$ and~$S$ are both non-empty~$K$-relations over~$X$, for
every~$X$-tuple~$t$ we have
  \begin{equation}
  R^*(t) = \frac{1}{N_R} R(t) = \frac{a}{aN_R} R(t)
  = \frac{b}{bN_S} S(t) =
  \frac{1}{N_S} S(t) = S^*(t),
\end{equation}
where the first equality follows from the fact that~$R^*$ is defined
from~$R$ through normalization, the second follows from the assumption
that~$a$ is non-zero, the third follows from the assumption
that~$aR = bS$, so, in particular,~$aR(t) = bS(t)$ and also~$R'=S'$
and~$a N_R = b N_S$, the fourth follows from the assumption that~$b$ is
non-zero, and the last follows from the definition of~$S^*$ through
normalization. This observation motivates the following definition of
the equivalence relation~$\equiv$ for arbitrary positive semirings
that are not necessarily semifields.

\paragraph{Definition of the Equivalence Relation for Positive Semirings}
Let~$K$ be a positive semiring. Two~$K$-relations~$R$ and~$S$ over the
same set of attributes are \emph{equivalent up to normalization},
denoted by~$R \equiv S$, if there exist non-zero~$a$ and~$b$ in~$K$
such that~$aR = bS$. It is obvious that~$\equiv$ is reflexive and
symmetric. The next lemma collects a few easy facts about~$\equiv$,
the first of which states that~$\equiv$ is also transitive, and hence
an equivalence relation. We write~$[R]$ for the equivalence class
of~$R$ under~$\equiv$.

\begin{lemma} \label{lem:easyfacts2} Let~$R(X),S(X),T(X)$
  be~$K$-relations over the same set~$X$ of attributes. The following statements hold:
\begin{enumerate} \itemsep=0pt
\item If $R \equiv S$ and $S \equiv T$, then $R \equiv T$.
\item If $R \equiv S$, then $R' = S'$.
\item If $R \equiv S$ and $Y \subseteq X$, then $R[Y] \equiv S[Y]$.
\end{enumerate}
\end{lemma}

\begin{proof}
  For~1, assume that~$a R = b S$ and~$c S = d T$ for non-zero~$a$
  and~$b$ in~$K$, and non-zero~$c$ and~$d$ in~$K$. Since~$K$ has no
  zero-divisors we have that~$ac$ and~$bd$ are non-zero.
  Moreover,
  \begin{equation}
  ac R = ca R = cb S = bc S = bd T,
  \end{equation}
  where the first equality is commutativity, the second follows
  from~$a R = b S$, the third is commutativity, and the last follows
  from~$c S = d T$. For~2, assume that~$a R = b S$ for non-zero~$a$
  and~$b$, and that~$R(t) \not= 0$ for some~$X$-tuple~$t$.
  Then~$aR(t) \not= 0$ because~$K$ has no zero-divisors,
  hence~$bS(t) \not= 0$ by the assumption that~$aR = bS$,
  and~$S(t) \not= 0$ since~$0$ annihilates~$K$. This
  shows~$R' \subseteq S'$ and the reverse inclusion follows from
  symmetry. For~3, assume that~$a R = b S$ for non-zero~$a$ and~$b$
  and that~$Y \subseteq X$. For every~$Y$-tuple~$u$ we have
  \begin{equation}
  a R(u) = a \sum_{r \in R': \atop r[Y] = u} R(r) =
  \sum_{r \in R' : \atop r[Y] = u} aR(r) =
  \sum_{r \in R' : \atop r[Y] = u} bS(r) =
  \sum_{s \in S' : \atop s[Y] = u} bS(r) =
  b \sum_{s \in S' : \atop s[Y] = u} S(r) =
  b S(u),
  \end{equation}
  where the first equality follows from~\eqref{eqn:marginal}, the second
  is distributivity, the third follows from the assumption
  that~$aR = bS$, the fourth follows from point 2 in this lemma, the
  fifth is again distributivity, and the sixth is~\eqref{eqn:marginal}.
\end{proof}

\section{Consistency of Two $K$-Relations} \label{sec:tworelations}

For ordinary relations~$R(X)$ and~$S(Y)$, there are several different
ways to define the concept of~$R$ and~$S$ being \emph{consistent},
and all these concepts turn out to be  equivalent to each other. One way is to say that~$R$
and~$S$ arise as the projections~$T[X]$ and~$T[Y]$ of a single
relation~$T$ over the union of attributes~$XY$. Another way is to say
that~$R$ and~$S$ agree on their projections to the set~$Z = X \cap Y$ of their common
attributes. Yet a third way is to say that their
ordinary join~$R \Join S$ projects to~$R$ on~$X$ and to~$S$ on~$Y$. In
this section, we study the analogous concepts for~$K$-relations with
consistency defined up to normalization. Along the way, we will also
define a notion of~$\Join$ for two~$K$-relations.

\subsection{Consistency of Two $K$-Relations and Their Join}

Let~$K$ be an arbitrary but fixed  positive semiring. We start with the
definition of consistency up to normalization, or more simply,
consistency of two~$K$-relations.

\paragraph{Consistency Up to Normalization} Let~$R(X)$ and~$S(Y)$ be
two~$K$-relations. We say that~$R$ and~$S$ are \emph{consistent} if
there is a~$K$-relation~$T(XY)$ such that~$R \equiv T[X]$
and~$S \equiv T[Y]$. We say that~$T$ \emph{witnesses} their
consistency. Two equivalence classes~$[R]$ and~$[S]$ of~$K$-relations
are \emph{consistent} if their representatives~$R$ and~$S$ are consistent. It
is easy to see that this notion of consistency among equivalence
classes is well-defined in that it does not depend on the chosen
representatives~$R$ and~$S$. Indeed, if~$[R]$ and~$[S]$ are consistent
and~$T$ witnesses the consistency of~$R$ and~$S$, then for
every~$R_0 \equiv R$ and every~$S_0 \equiv S$, we have that~$T$ also
witnesses the consistency of~$R_0$ and~$S_0$, by the
transitivity of $\equiv$. Conversely, if~$T_0$ witnesses the consistency
of~$R_0 \equiv R$ and~$S_0 \equiv S$, then it also witnesses the
consistency of~$R$ and~$S$, again by the transitivity of $\equiv$.

\paragraph{Naive Join Operation of Two~$K$-Relations}
We want to define a join operation~$R \Join S$ for~$K$-relations~$R$
and~$S$ with the property that if~$R$ and~$S$ are consistent, then
their join witnesses the consistency. A natural candidate for such an
operation would be to define~$(R \Join S)(t)$ by~$R(t[X]) S(t[Y])$ for
every~$XY$-tuple~$t$, where~$X$ and~$Y$ are the sets of attributes
of~$R$ and~$S$, respectively. This is the straightforward
generalization of the ordinary join of ordinary relations since for
the Boolean semiring both definitions give the same
operation. Moreover, this is the way the join of bags is defined in SQL (see \cite{DBLP:books/daglib/0011318}).
As we show below, however, this naive generalization does
not work: in fact, even for bags, the bag defined this way does not always
witness the consistency of  two consistent bags.

\begin{example} \label{ex:naivejoin}
Let~$R(AB),S(BC),J(ABC),U(ABC)$ be the four bags given by the
following tables of multiplicities (the $\#$-column is the multiplicity):
  \begin{center}
  \begin{tabular}{lllllllllllllllllll}
  $R(AB)$\, \# & \;\;\;\; & $S(BC)$ \# & \;\;\;\; & $J(ABC)$\, \# & \;\;\;\; & $U(ABC)$\, \# \\
  \;\;\;\, 1 2 : 6 & & \;\;\; 2 3 : 2 & & \;\;\, 1 2 3 : 12 & & \;\;\;\, 1 2 3 : 6 \\
  \;\;\;\, 2 3 : 3 & & \;\;\; 2 4 : 2 & & \;\;\, 1 2 4 : 12 & & \;\;\;\, 1 2 4 : 6 \\
  \;\;         & & \;\;\; 3 4 : 2 & & \;\;\, 2 3 4 : 6 & & \;\;\;\, 2 3 4 : 6
  \end{tabular}
  \end{center}
  Consider also the marginals of $J$ and $U$ on $AB$ and $BC$:
  \begin{center}
  \begin{tabular}{llllllllllllllllll}
  $J[AB]$\; \# & \;\;\;\; & $J[BC]$\; \# & \;\;\;\; & $U[AB]$\; \# & \;\;\;\; & $U[BC]$\, \# \\
  \;\;\, 1 2 : 24 & & \;\;\, 2 3 : 12 & & \;\;\; 1 2 : 12 & & \;\;\; 2 3 : 6 \\
  \;\;\, 2 3 : 6 & & \;\;\, 2 4 : 12 & & \;\;\; 2 3 : 6 & & \;\;\; 2 4 : 6 \\
  \;\;         & & \;\;\, 3 4 : 6 & & \;\;\; & & \;\;\; 3 4 : 6
  \end{tabular}
  \end{center}
  The bags~$R$ and~$S$ are consistent since~$U$ witnesses their
  consistency: $2R = U[AB]$ and~$3S = U[BC]$. The bag~$J$ is
  actually the naive join of~$R$ and~$S$ defined
  by~$J(t) = R(t[AB])S(t[BC])$, and there are no non-zero~$a$ and~$b$
  in~$\mathbb{N}$ such that~$aR = bJ[AB]$, and also there are no
  non-zero~$c$ and~$d$ in~$\mathbb{N}$ such that~$c S = dJ[BC]$.
\end{example}

\noindent  This example has shown that the \emph{naive} join need not witness the
consistency of~$R$ and~$S$. We need a different way of defining the
join operation.

\paragraph{Derivation of the New Join Operation}
To arrive at the definition of the join operation that will work for
arbitrary semirings, we turn again to probability distributions
from Example~\ref{ex:examplessemirings} as the motivating example.
Recall that a probability distribution is
a~$\mathbb{R}^{\geq 0}$-relation that
satisfies~\eqref{eqn:normalizationequation}.  As in the discussion for
defining the equivalence relation, this motivation will generalize to
any semifield beyond~$\mathbb{R}^{\geq 0}$. {From} there, generalizing
the definition to arbitrary semirings will be a small step.

Let~$R(X)$ and~$S(Y)$ be probability distributions. It is easy to see
that if~$X$ and~$Y$ were disjoint, then the~$\mathbb{R}^{\geq 0}$
relation given by~$t \mapsto R(t[X])S(t[Y])$ would again be a
probability distribution. This, however, fails badly if~$X$ and~$Y$
are not disjoint as can be seen from turning the example bags~$R(AB)$
and~$S(BC)$ from Example~\ref{ex:naivejoin} into probability
distributions through normalization (when seen
as~$\mathbb{R}^{\geq 0}$-relations). The catch is of course that
if~$Z = X \cap Y$ is non-empty, then two independent samples from the
distributions~$R$ and~$S$ need not agree on their projections
on~$Z$. The solution is to define the join~$R \Join S$ as the
probability distribution on~$XY$ that is sampled by the following
different process: first sample~$r$ from the distribution~$R$, then
sample~$s$ from the distribution~$S$ \emph{conditioned}
on~$s[Z] = r[Z]$, finally output the tuple~$rs$ which is well-defined
since~$s[Z] = r[Z]$.  This leads to the  expression
\begin{equation}
  (R \Join_{\mathrm{P}} S)(t) := R(t[X])S(t[Y])/S(t[Z]) \label{eqn:prop1}
\end{equation}
defined for all~$XY$-tuples~$t$, with convention that~$0/0 = 0$.
Observe that, by writing~$r := t[X]$,~$s :=
t[Y]$,~$u := t[Y\setminus Z]$ and~$v := t[Z]$, the factor~$R(t[X])$
in~\eqref{eqn:prop1} is the probability~$R(r)$ of getting~$r$ in a
sample from the distribution~$R$, and the factor~$S(t[Y])/S(t[Z])$ is
the probability~$S(s)/S(v) = S(uv)/S(v)$ of getting~$s$ in a sample
from the distribution~$S$ conditioned on~$s[Z] = v = r[Z]$.

Naturally, we could have equally well considered the reverse sampling
process that first samples~$s$ from~$S$, and then samples~$r$ from~$R$
\emph{conditioned} on~$r[Z] = s[Z]$. This would lead to the alternative expression
\begin{equation}
(R \mathbin{{_\mathrm{P}}\!\!\Join} S)(t) :=
S(t[X])R(t[Y])/R(t[Z]). \label{eqn:prop2}
\end{equation}
It is clear from the definitions
that~$R \Join_{\mathrm{P}} S = S \mathbin{{_\mathrm{P}}\!\!\Join} R$,
but for the two proposals to agree we would need to
have~$R[Z] = S[Z]$. Luckily, this can actually be seen to hold in
case~$R$ and~$S$ are consistent probability distributions since if~$T$
is a~$\mathbb{R}^{\geq 0}$-relation that satisfies~$T[X] = R$
and~$T[Y] = S$, then also~$T[Z] = R[Z] = S[Z]$ by Part~3 of
Lemma~\ref{lem:easyfacts1}. It will follow from the lemmas below that
in such a case we also have~$(R\Join S)[X] \equiv R$
and~$(R \Join S)[Y] \equiv S$ for both~$\Join = \Join_{\mathrm{P}}$
and~$\Join = \mathbin{{_\mathrm{P}}\!\!\Join}$, which is what we want.

It is clear that the expression in~\eqref{eqn:prop1}, in addition to
being defined for~$K = \mathbb{R}^{\geq 0}$, could have been defined
for any semiring~$K$ that is actually a semifield where a division
operation is available. On the other hand,  to obtain an expression
that works for arbitrary semirings, we need to eliminate the
divisions. A natural approach for this would be to multiply the
expression in~\eqref{eqn:prop1} by the
product~$\prod_{s \in S[Z]'} S(s[Z])$ of all the values that appear in
the denominator. This way, the denominator would cancel, yet the
resulting~$K$-relation would remain equivalent up to normalization
because the multiplying products do not depend on the tuple~$t$. We
are now ready to formally define this.

\paragraph{Definition of the Join of Two $K$-Relations}
For a~$K$-relation~$T(X)$, a
subset~$Z \subseteq X$, and a~$Z$-tuple~$u$, define
\begin{equation}
c^*_{T,Z} := \prod_{v \in T[Z]'} T(v) \;\;\;\text{ and }\;\;\;
c_{T}(u) := \prod_{v \in T[Z]': \atop v \not= u} T(v), \label{eqn:cr}
\end{equation}
with the understanding that the empty product evaluates to~$1$, the
unit of the semiring~$K$. Observe that~$c_{T}(u)$ and~$c^*_{T,Z}$ are
always non-zero because~$T[Z]'$ is precisely the set of~$Z$-tuples~$v$
with non-zero~$T(v)$, and~$K$ has no zero-divisors. The \emph{join} of
two~$K$-relations~$R(X)$ and~$S(Y)$ is the~$K$-relation over~$XY$ defined, for every~$XY$-tuple~$t$, by
\begin{equation}
 (R \Join S)(t) :=
 R(t[X]) S(t[Y]) c_S(t[X \cap Y]). \label{eqn:joinasymmetric}
\end{equation}
It is worth noting at this point that the
identity~$c^*_{T,Z} = c_T(u) T(u)$ holds, which means that
whenever~$K$ is a semifield such as~$\mathbb{R}^{\geq 0}$, we
have~$c^*_{T,Z}/T(u) = c_T(u)$ for all~$u \in T[Z]'$, and therefore
\begin{equation}
  R \Join S = c^*_{S[Z]} (R \Join_{\mathrm{P}} S), \label{eqn:coincides}
\end{equation}
where $\Join_{\mathrm{P}}$ is defined as in~\eqref{eqn:prop1}.

Note that, for ordinary relations, the join operation just introduced coincides with the (ordinary) join operation in relational databases.
Note also that the definition
of~$\Join$ is asymmetric. Thus, on the face of its definition,  the join of two~$K$-relations
need not be commutative, i.e., there may be $K$-relations $R$ and $S$ such that $R\Join S \not = S\Join R$. As a matter of fact, something  stronger holds: in general, $R\Join S \not \equiv S\Join R$; furthermore, as we shall see next, this happens even for bags.
Nonetheless,  we will show later
that~$R \Join S \equiv S \Join R$ does hold in case~$R$ and~$S$ agree
on their common marginals; moreover, in that case,  both joins $R\Join S$ and $S\Join R$  witness
the consistency of $R$ and $S$.

The example that follows illustrates the definition of the join and also shows that the join operation need not be commutative, not even up to equivalence.

\begin{example}
  Consider the bags~$R(AB)$ and~$S(BC)$ from
  Example~\ref{ex:naivejoin}, where they were shown to be consistent using the
  bag~$U(ABC)$ from the same example. The joins~$V := R \Join S$
  and~$W := S \Join R$ defined by~\eqref{eqn:joinasymmetric} are the
  bags given by the following tables of multiplicities. We display~$V$
  and~$W$ alongside their  marginals on~$AB$ and~$BC$.
  \begin{center}
  \begin{tabular}{lllllllllllllllllllllll}
  $V(ABC)$\; \# & \; & $W(ABC)$\; \# & \; & $V[AB]$\; \# & \; & $V[BC]$\; \# & \; & $W[AB]$\; \# & \; & $W[BC]$\; \# \\
  \;\;\; 1 2 3 : 24 & & \;\;\;\; 1 2 3 : 36 & & \;\;\; 1 2 : 48 & & \;\;\; 2 3 : 24 & & \;\;\;\; 1 2 : 72 & & \;\;\;\; 2 3 : 36 \\
  \;\;\; 1 2 4 : 24 & & \;\;\;\; 1 2 4 : 36 & & \;\;\; 2 3 : 24 & & \;\;\; 2 4 : 24 & & \;\;\;\; 2 3 : 36 & & \;\;\;\; 2 4 : 36 \\
  \;\;\; 2 3 4 : 24 & & \;\;\;\; 2 3 4 : 36 & & \;\;\;    & & \;\;\; 3 4 : 24 & & \;\;\;    & & \;\;\;\; 3 4 : 36
  \end{tabular}
  \end{center}
  For example, the entry~$V(234)$ is computed
  as~$3\!\cdot\!2\!\cdot\!(2\!+\!2) = 24$
  according to the expression~\eqref{eqn:joinasymmetric} that defines the
  join operation. By inspection, we have that~$8R=V[AB]$
  and~$12S=V[BC]$, so~$V$ witnesses the consistency of~$R$ and~$S$,
  and~$24R = W[AB]$ and~$18S = W[BC]$, so~$W$ also witnesses the
  consistency of~$R$ and~$S$. Indeed,~$3V = 2W$, which shows that  $R\Join S \equiv S\Join R$ holds for these two bags $R$ and $S$.

  Next, we use the bag~$R$ and another bag~$T$ to show that~$\Join$ need not be commutative, not
  even up to equivalence.  Consider the bag~$T(BC)$ given below by its
  table of multiplicities together with~$J_1 = R \Join T$
  and~$J_2 = T \Join R$:
  \begin{center}
  \begin{tabular}{lllllllllllllllllllllll}
  $T(BC)$\, \# & \; & $J_1(ABC)$\; \# & \; & $J_2(ABC)$\; \# \\
  \;\;\;\, 2 3 : 2 & & \;\;\;\; 1 2 3 : 48 & & \;\;\;\; 1 2 3 : 36 \\
  \;\;\;\, 2 4 : 2 & & \;\;\;\; 1 2 4 : 48 & & \;\;\;\; 1 2 4 : 36 \\
  \;\;\;\, 3 4 : 4 & & \;\;\;\; 2 3 4 : 48 & & \;\;\;\; 2 3 4 : 72
  \end{tabular}
\end{center}
Clearly, there are no non-zero~$a$ and~$b$ in~$\mathbb{N}$ such
that~$a J_1 = b J_2$, thus~$R \Join T \not\equiv T \Join R$. Note that
the pair of~$K$-relations~$R$ and~$T$ that gave this
has~$R[B] \not\equiv T[B]$. This is no coincidence, since, in what
follows, we will show that if the two~$K$-relations have
equivalent common marginals, then their join is commutative up to
equivalence. This was the case, for example, for the pair of bags~$R$
and~$S$ considered also in this example, which had~$R[B] \equiv S[B]$
and~$R \Join S \equiv S \Join R$.
\end{example}

\paragraph{Properties of the Join Operation}
The first property we show about the join of two~$K$-relations is that it
is well-defined in the sense that its equivalence class does not
depend on the representatives. In other words, we show that the join operation $\Join$ is a congruence with respect to the equivalence relation
$\equiv$ on $K$-relations.

\begin{lemma} Let $R, R_0$ be two $K$-relations over a set $X$ and let~$S, S_0$ be two $K$-relations over a set $Y$. If~$R \equiv R_0$
  and~$S \equiv S_0$, then~$R \Join S \equiv R_0 \Join S_0$.
\end{lemma}

\begin{proof}
  Let~$X$ be the set of attributes of~$R$ and~$R_0$, and let~$Y$ be
  that of~$S$ and~$S_0$. Write~$Z = X \cap Y$.  Let~$a$ and~$b$ be
  non-zero elements in~$K$ such that~$a R = b R_0$, and let~$c$
  and~$d$ be non-zero elements in~$K$ such that~$c S = d S_0$.  First
  note that~$R[Z]' = R_0[Z]'$ by Parts~3 and~2 in
  Lemma~\ref{lem:easyfacts2}.  Let~$m$ be the cardinality
  of~$R[Z]' = R_0[Z]'$ and set~$a^* = a c^m$ and~$b^* = b d^m$.  We
  argue that~$a^* (R \Join S)(t) = b^* (R_0 \Join S_0)(t)$ for
  every~$XY$-tuple~$t$. Fix an~$XY$-tuple~$t$ and
  distinguish the cases~$t[Z] \not\in R[Z]'$ from~$t[Z] \in R[Z]'$. In
  the first case, we have~$R(t[Z])=0$ and hence~$R(t[X]) = 0$ by Part~2
  in Lemma~\ref{lem:easyfacts1}, so~$R_0(t[X]) = 0$ by Part~2 in
  Lemma~\ref{lem:easyfacts2}. It follows
  that~$(R \Join S)(t) = (R_0 \Join S_0)(t) = 0$ in this case. In the
  second case, we have~$m \geq 1$ and
  \begin{align}
  a^* (R\Join S)(t) = a R(t[X]) \cdot c S(t[Y]) \cdot \prod_{r \in R[Z]' :
  \atop r \not= t[Z]} c S(r) \label{eqn:onehand}
  \end{align}
  on one hand since $a^* = acc^{m-1}$, and
  \begin{align}
  b^* (R_0 \Join S_0)(t) = b R_0(t[X]) \cdot d S_0(t[Y]) \cdot
  \prod_{r \in R_0[Z]': \atop r \not= t[Z]} d S_0(r), \label{eqn:theother}
  \end{align}
  on the other since~$b^* = bdd^{m-1}$. The right-hand sides
  of~\eqref{eqn:onehand}~and~\eqref{eqn:theother} are equal
  by~$a R = b R_0$ and~$c S = d S_0$, and~$R[Z]'=R_0[Z]'$, so the
  lemma is proved.
\end{proof}

Next, we show that the support of the join is the ordinary join
of the supports.

\begin{lemma} \label{lem:supportsjoin}
  For all $K$-relations $R$ and $S$, we have that
  $(R \Join S)' = R' \Join S'$.
\end{lemma}

\begin{proof}
  Let~$X$ be the set of attributes of~$R$, and let~$Y$ be that
  of~$S$. Write~$Z = X \cap Y$ and~$T = R \Join S$. Fix
  an~$XY$-tuple~$t$. If~$t$ is in~$T'$, then~$T(t) \not= 0$ and
  in particular~$R(t[X]) \not= 0$ and~$S(t[Y]) \not= 0$
  by~\eqref{eqn:joinasymmetric}. It follows that~$t[X]$ is in~$R'$
  and~$t[Y]$ is in~$S'$; i.e.,~$t$ is in the relational join of~$R'$
  and~$S'$. Conversely, if~$T(t) = 0$, then
  by~\eqref{eqn:joinasymmetric} again either~$R(t[X])=0$
  or~$S(t[Y])=0$ or~$c_S(t[Z]) = 0$ since~$K$ has no
  zero-divisors. The third case is absurd: we already argued
  that~$c_S(t[Z]) \not= 0$ since~$S[Z]'$ is precisely the set
  of~$Z$-tuples~$v$ with~$S(v) \not= 0$. In the first two cases, we can
  conclude that either~$t[X]$ is not in~$R'$ or~$t[Y]$ is not in~$S'$,
  so~$t$ is not in their join.
\end{proof}

The \emph{left semijoin}~$R' \ltimes S'$ of two ordinary
relations~$R'(X)$ and~$S'(Y)$ is the set of~$X$-tuples
in~$R'$ that join with some~$Y$-tuple in~$S'$,
i.e.,~$R' \ltimes S' = (R' \Join S')[X]$. We use
Lemma~\ref{lem:supportsjoin} to show that the asymmetric join behaves
like a left semijoin up to equivalence, in a strong
sense~(with~$a = 1$).

\begin{lemma} \label{lem:semijoinlike} For all $K$-relations $R$ and
  $S$ and all $r \in R' \ltimes S'$, we have that
  $(R \Join S)(r) = c^*_S R(r)$.
\end{lemma}

\begin{proof}
  Let~$X$ and~$Y$ be the sets of attributes of~$R$ and~$S$,
  respectively, and write~$Z = X \cap Y$ and~$T = R \Join S$.  Fix
  an~$X$-tuple~$r \in R' \ltimes S'$ and write~$u = r[Z]$. We have
  \begin{equation}
    T(r) = \sum_{t \in T': \atop t[X]=r} T(t) =
    \sum_{t \in T':\atop t[X]=r} R(t[X])S(t[Y]) c_S(t[Z]) =
    c_S(u) R(r) \sum_{t \in T': \atop t[X]=r} S(t[Y]), \label{eqn:firsthah}
   \end{equation}
   where the first equality follows from~\eqref{eqn:marginal}, the
   second follows from~\eqref{eqn:joinasymmetric}, and the third
   follows from the condition that~$t[X]=r$ because~$Z \subseteq X$
   implies~$t[Z]=t[X][Z]=r[Z]=u$. At this point, we use the fact
   that~$r \in R' \ltimes S'$ and hence~$r \in R'$, together with
   Lemma~\ref{lem:supportsjoin}, to argue that the map
 \begin{align}
   f : & \{ t \in T' : t[X] = r \} \rightarrow
   \{ s \in S' : s[Z] = u \} :: t \mapsto t[Y]
 \label{eqn:bijectionh}
 \end{align}
 is a bijection.  Indeed, since by Lemma~\ref{lem:supportsjoin}
 each~$t \in T'$ comes from the relational join of~$R'$ and~$S'$, for
 each~$t \in T'$ such that~$t[X]=r$ there exists~$s \in S'$
 with~$t[Y]=s$ and~$s[Z]=t[Y][Z]=t[Z]=t[X][Z]=r[Z]=u$. Clearly
 this~$s=t[Y]$ is uniquely determined from~$t$. Conversely,
 if~$s \in S'$ is such that~$s[Z]=u=r[Z]$, then the join tuple~$t$
 of~$r$ and~$s$ exists, it is in~$T'$ by Lemma~\ref{lem:supportsjoin}
 and the fact that~$r \in R'$, and moreover~$t[X] =
 r$. This~$t$ is uniquely determined from~$s$ (and the
 fixed~$r$). This proves that~\eqref{eqn:bijectionh} is a
 bijection. Therefore, continuing from~\eqref{eqn:firsthah}, we have
 \begin{equation}
 c_S(u) R(r) \sum_{t \in T': \atop t[X]=r} S(t[Y]) = c_S(u) R(r) \sum_{s \in S': \atop s[Z]=u} S(s) = c_S(u) R(r) S(u). \label{eqn:last}
\end{equation}
where the first equality follows from the just shown fact
that~\eqref{eqn:bijectionh} is a bijection, and the second follows
from~\eqref{eqn:marginal}. Recall now  that~$u = r[Z]$
and~$r \in R' \ltimes S'$, which means that~$u \in (R' \Join S')[Z]$.
In particular,~$u \in S'[Z]$, so~$u \in S[Z]'$ by Part~2 of
Lemma~\ref{lem:easyfacts1}.  Thus,
by~\eqref{eqn:cr}, we have ~$c^*_S = c_S(u) S(u)$, and
equations~\eqref{eqn:firsthah} and~\eqref{eqn:last} actually show
that~$T(r) = c^*_S R(r)$.
\end{proof}

Next we show that if two~$K$-relations are consistent in the sense
that their marginals on the common attributes are equivalent, then
their join commutes up to equivalence. Later we will use this to
argue that this sense of consistency in terms of marginals is
equivalent to the one defined earlier in this section, and thus that
if two~$K$-relations are consistent, then their join commutes up
to equivalence.

\begin{lemma} \label{lem:commutativeifconsistent} For
  all~$K$-relations~$R(X)$ and~$S(Y)$,
  if~$R[X \cap Y] \equiv S[X \cap Y]$, then~$\Join$ commutes on~$R$
  and~$S$ up to equivalence, i.e.,~$R \Join S \equiv S \Join R$.
\end{lemma}

\begin{proof}
  Write~$Z = X \cap Y$.  Let~$a$ and~$b$ be non-zero and such
  that~$a R[Z]= b S[Z]$.  First note that~$R[Z]' = S[Z]'$ by Part~3
  and~2 of Lemma~\ref{lem:easyfacts2}. Let~$m$ be the cardinality of
  $R[Z]' = S[Z]'$.  If~$m = 0$,
  then~$(R \Join S)(t) = R(t[X])S(t[Y]) = S(t[Y])R(t[X]) = (S \Join
  R)(t)$ for every~$XY$-tuple~$t$, and we are done. Assume then
  that~$m \geq 1$ and set~$a^* = a^{m-1}$ and~$b^* = b^{m-1}$.  We
  argue that~$b^* (R \Join S)(t) = a^* (S \Join R)(t)$ for
  every~$XY$-tuple~$t$. Fix an~$XY$-tuple~$t$ and
  distinguish the cases~$t[Z] \not\in S[Z]'$ from~$t[Z] \in S[Z]'$. In
  the first case we have~$S(t[Y]) = 0$ by Part~2 of
  Lemma~\ref{lem:easyfacts1} and it follows
  that~$b^* (R \Join S)(t) = 0 = a^* (S \Join R)(t)$ in this case. In
  the second case we have
  \begin{equation}
  b^* (R \Join S)(t) = R(t[X])S(t[Y])\prod_{s \in S[Z]': \atop s\not=t[Z]}
  b S(r) \label{eqn:onehandn}
  \end{equation}
  on one hand since $b^* = b^{m-1}$ and $t[Z] \in S[Z]'$, and
  \begin{equation}
  a^* (S \Join R)(t) = S(t[Y])R(t[X])\prod_{r \in R[Z]': \atop r\not=t[Z]}
  a R(r) \label{eqn:otherhandn}
  \end{equation}
  on the other since~$a^* = a^{m-1}$ and~$t[Z] \in S[Z]'=R[Z]'$.  Now,
  given hat~$aR(r) = bS(r)$ for every~$Z$-tuple~$r$, the right-hand
  sides of~\eqref{eqn:onehandn} and~\eqref{eqn:otherhandn} are equal,
  and the lemma is proved.
\end{proof}

We are ready to show that the join witnesses the consistency of any two
consistent~$K$-relations. Along the way, we also prove that
two~$K$-relations are consistent if and only if their marginals on the
common attributes are equivalent. This result tells that the join operation on two $K$-relations introduced here possesses most of the desirable properties that the join of ordinary relations in relational databases does.

\begin{lemma} \label{lem:characterization} Let~$R(X)$ and~$S(Y)$
  be~$K$-relations. The following statements are equivalent:
  \begin{enumerate} \itemsep=0pt
  \item[(a)] $R$ and $S$ are consistent.
  \item[(b)] $R[X \cap Y] \equiv S[X \cap Y]$.
  \item[(c)] $R'$ and $S'$ are consistent and $R \Join S \equiv S \Join R$.
  \item[(d)] $R \equiv (R \Join S)[X]$ and $S \equiv (R \Join S)[Y]$.
  \end{enumerate}
\end{lemma}

\begin{proof} Write~$Z = X \cap Y$. For~(a) implies~(b), let~$T$
  witness that~$R$ and~$S$ are consistent, so~$R \equiv T[X]$
  and~$S \equiv T[Y]$. Then, by Part~3 of Lemma~\ref{lem:easyfacts2},
  we have~$R[Z] \equiv T[X][Z]$ and~$S[Z] \equiv T[Y][Z]$. Since by
  Part~3 of Lemma~\ref{lem:easyfacts1} we also
  have~$T[X][Z] = T[Z] = T[Y][Z]$, we get~$R[Z] \equiv S[Z]$, as was
  to be shown. For~(b) implies~(c) first apply Part~2 of
  Lemma~\ref{lem:easyfacts1} followed by Part~2 of
  Lemma~\ref{lem:easyfacts2} to conclude
  that~$R'[X \cap Y] = S'[X \cap Y]$ and hence that~$R'$ and~$S'$ are
  consistent as ordinary relations. By
  Lemma~\ref{lem:commutativeifconsistent} we also
  have~$R \Join S \equiv S \Join R$. For~(c) implies~(d) first note
  that the consistency of~$R'$ and~$S'$ implies
  that~$R' = R' \ltimes S'$ and~$S' = S' \ltimes R'$. Thus,
  Lemma~\ref{lem:semijoinlike} gives~$R \equiv (R \Join S)[X]$
  and~$S \equiv (S \Join R)[Y]$. Together with the assumption
  that~$R \Join S \equiv S \Join R$ this also
  gives~$S \equiv (R \Join S)[Y]$ by Part~3 of
  Lemma~\ref{lem:easyfacts2}. That~(d) implies~(a) is direct
  since~(d) says that~$R \Join S$ witnesses the consistency of~$R$
  and~$S$.
\end{proof}

\subsection{Justification of the Join of Two $K$-Relations}

In this section, we address the question whether the join operation on
two relations that we defined in Section~\ref{sec:tworelations} is
well motivated. For the rest of this section, fix a finite set of
attributes and let~$\tuples$ denote the set of all tuples over these
attributes, which we assume is a computable set through the
appropriate encodings. We also assume that the positive
semiring~$K$ is a computable structure in the sense that the elements
of its domain admit a computable presentation that makes its
operations be computable functions. The bag semiring~$\mathbb{N}$, as
well as the semiring~$\mathbb{Q}^{\geq 0}$ of non-negative rationals
and many others, are of course computable in this sense. Furthermore, we
require the equivalence relation~$\equiv$ to be decidable; in other
words, we require that the following computational problem is
decidable:
\begin{center}
\emph{Given two~$K$-relations~$R$ and~$S$ over the same set,
does~$R \equiv S$ hold?}
\end{center}
We note that for the bag semiring~$\mathbb{N}$, as well as for the
semiring~$\mathbb{Q}^{\geq 0}$ of non-negative rationals, this problem
is very easily decidable, even polynomial-time solvable through what
we call the \emph{ratio test}: first, check
whether~$R' = S'$, and then check  whether~$R(t_1)/S(t_1) =
R(t_2)/S(t_2)$ holds for every two tuples~$t_1$ and~$t_2$ in~$R' =
S'$.

\paragraph{Deciding Consistency Despite the Plethora of Witnesses}
Let~$R$ and~$S$ denote two~$K$-relations on the sets of attributes~$X$
and~$Y$ and consider the following computational problem:
\begin{center}
  \emph{Given two~$K$-relations~$R$ and~$S$, are~$R$ and~$S$
    consistent?}
\end{center}
For an infinite positive semiring~$K$, such as the bag
semiring~$\mathbb{N}$, there is no immediate and a priori reason to
think that this problem is algorithmically solvable. The difficulty is
that in principle there are infinitely many candidate~$K$-relations to
test as witness for consistency, and the arithmetic theory of the
natural numbers is highly undecidable. However, what
Lemma~\ref{lem:characterization} shows is that the two
given~$K$-relations~$R$ and~$S$ are consistent if and only if~the
single, finite and explicitly defined~$K$-relation~given
by~$R \Join S$ witnesses their consistency. Thus, if~$K$ is a semiring
for which the equivalence relation~$\equiv$ is decidable, this can be
checked in finite time and the problem is decidable. In the next
example, we show that, even for bags, the consistency of two bags may very well be witnessed by infinitely
many pairwise inequivalent witnesses.

\begin{example}
  Let~$a$ be a positive integer and let~$R(AB)$,~$S(BC)$
  and~$T_a(ABC)$ be the three bags given by the following multiplicity
  tables, listed alongside the two projections of~$T_a$ on~$AB$
  and~$BC$:
  \begin{center}
  \begin{tabular}{lllllllllllllllllll}
  $R(AB)$\, \# & \; & $S(BC)$ \# & \; & $T_a(ABC)$\, \# & \; & $T_a[AB]$\, \# & \; & $T_a[BC]$\, \# \\
  \;\;\;\, 0 0 : 1 & & \;\;\; 0 0 : 1 & & \;\;\;\; 0 0 0 : $a$ & & \;\;\;\; 0 0 : $a+1$ & & \;\;\;\; 0 0 : $a+1$ \\
  \;\;\;\, 1 0 : 1 & & \;\;\; 0 1 : 1 & & \;\;\;\; 0 0 1 : 1 & & \;\;\;\; 1 0 : $a+1$ & & \;\;\;\; 0 1 : $a+1$ \\
  \;\;\;\,   & & \;\;\;  & & \;\;\;\; 1 0 0 : 1 & & \;\;\;\;  & & \;\;\;\; \\
  \;\;\;\,   & & \;\;\;  & & \;\;\;\; 1 0 1 : $a$ & & \;\;\;\;  & & \;\;\;\;
  \end{tabular}
  \end{center}
  It is evident that~$T_a[AB] = (a+1)R$ and~$T_a[BC] = (a+1)S$,
  but~$T_a \not\equiv T_b$ unless~$a = b$. The conclusion is that
  there are infinitely many different equivalence classes that witness
  the consistency of~$R$ and~$S$.
%
%
\end{example}

\paragraph{Entropy Maximization}
Let us turn our attention again to probability distributions. The
canonical representatives of the equivalence classes are
the~$\mathbb{R}^{\geq 0}$-relations~$T$ that
satisfy~\eqref{eqn:normalizationequation}. We argued already that for
such canonical~$\mathbb{R}^{\geq 0}$-relations we have that~$\equiv$
agrees with~$=$. Therefore, the set of
canonical~$\mathbb{R}^{\geq 0}$-relations~$T$ that witness the
consistency of two given probability distributions~$R(X)$ and~$S(Y)$
can be identified with the set of feasible solutions of a linear
program that has one real variable~$x_t$ representing~$T(t)$ for
each~$XY$-tuple~$t$ in the join of the supports of~$R$ and~$S$:
\begin{equation}
\begin{array}{lll}
\sum_{t: t[X] = r} x_t = R(r) & \;\;\; & \text{ for each } r \in R', \\
\sum_{t: t[Y] = s} x_t = S(r) & \;\;\; & \text{ for each } s \in S', \\
\sum_{t} x_t = 1 & & \\
x_t \geq 0 & & \text{ for each } t \in R' \Join S'. \\
\end{array} \label{eqn:linearprogram}
\end{equation}
The set of probability distributions~$P(XY)$ that witness the
consistency of~$R$ and~$S$ is thus a polytope~$\mathrm{W}(R,S)$ which
is non-empty if and only if~$R$ and~$S$ are consistent. A natural
question to ask is whether there is some particular probability
distribution in this polytope that is better motivated than any
other such probability distribution. For example, following the principle of maximum entropy, we
could ask for the probability distribution that maximizes
\emph{Shannon's Entropy} (see section~2.1. in~\cite{CoverThomas})
among those that witness the consistency, i.e., we want to maximize
\begin{equation}
  \Entropy_P(XY) = -\sum_{xy \in P(XY)'} P(xy) \log_2(P(xy)) \label{eqn:entropy}
\end{equation}
subject to the constraint that $P$ is in $\mathrm{W}(R,S)$.
Since the entropy is a concave function over the probability simplex
(Theorem~2.7.3 in~\cite{CoverThomas}), and since~$\mathrm{W}(R,S)$ is
a bounded polytope and hence a compact subset of~$\mathbb{R}^n$ in
appropriate dimension~$n$ (unless it is empty), the maximum
of~\eqref{eqn:entropy} exists and is achieved at a unique point
in~$\mathrm{W}(R,S)$. In our setting, writing~$Z = X \cap Y$, it is
perhaps more natural to maximize the \emph{conditional entropy}, i.e.,
\begin{equation}
  \Entropy_P(XY|Z) := -\sum_{z \in P(Z)'} P(z) \sum_{xy \in P(XY)'} P(xy|z) \log_2(P(xy|z)), \label{eqn:conditionalentropy}
\end{equation}
where $P(xy|z) := 0$ if $(xy)[Z] \not= z$ or $P(z) = 0$, and
$P(xy|z) := P(xy)/P(z)$ otherwise, with the added convention that
$0\log_2(0) = 0$.
For being a convex combination of concave functions over the
probability simplex, the conditional entropy~$\Entropy_P(XY| Z)$ is
again a concave function of~$P$ ranging over~$\mathrm{W}(R,S)$, which
means that the maximum also exists and is achieved at a unique point
in~$\mathrm{W}(R,S)$. We write~$R \Join_{\mathrm{H}} S$ for the unique
probability distribution in~$\mathrm{W}(R,S)$ that achieves the
maximum of~\eqref{eqn:entropy} and we write~$R \Join_{\mathrm{CH}} S$
for the one that achieves the maximum
of~\eqref{eqn:conditionalentropy}. Note that, a priori, due to the
logarithms in the definition of entropy, the probability
distributions~$R \Join_{\mathrm{H}} S$ and~$R \Join_{\mathrm{CH}} S$
need not even have rational components. Interestingly, as will follow
from the development below, our join operation~$\Join$ applied
to consistent probability distributions coincides with
both~$\Join_{\mathrm{H}}$ and~$\Join_{\mathrm{CH}}$, up to the
equivalence, which means that both~$R \Join_{\mathrm{H}} S$
and~$R \Join_{\mathrm{CH}} S$ \emph{are} indeed rational probability
distributions in case~$R$ and~$S$ are themselves rational.

We argued already in~\eqref{eqn:coincides} that, for probability
distributions~$R$ and~$S$, our join~$R \Join S$ coincides
with~$R \Join_{\mathrm{P}} S$ up to equivalence. Moreover, if~$R$
and~$S$ are consistent, then we have~$R[Z]=S[Z]$ for~$Z := X \cap Y$,
which means that if we write~$r := t[X]$,~$s := t[Y]$,~$u := t[Z]$,
and~$U := R[Z]=S[Z]$, then
\begin{equation}
(R \Join_{\mathrm{P}} S)(t) = R(r)S(s)/U(u). \label{eqn:identity}
\end{equation}
This identity implies that~$R \Join_{\mathrm{P}} S$ is a \emph{product
  extension} of~$R$ and~$S$ in the sense of
Malvestuto~\cite{DBLP:journals/dm/Malvestuto88} (see the first
paragraph of page~73 in \cite{DBLP:journals/dm/Malvestuto88}),
hence~$R \Join_{\mathrm{P}} S$ maximizes entropy as a consequence of
Malvestuto's Theorem~8. We reproduce his short proof for completeness.

\begin{lemma}[\cite{DBLP:journals/dm/Malvestuto88}] \label{lem:maximizesentropy}
If~$R$ and~$S$ are consistent probability distributions, then 
~$R \Join_{\mathrm{P}} S = R \Join_{\mathrm{H}} S$.
\end{lemma}

\begin{proof}
  Let~$X$ and~$Y$ be the sets of attributes of~$R$ and~$S$,
  write~$Z = X \cap Y$, and assume that~$R$ and~$S$ are consistent.
  Let~$U := R[Z] = S[Z]$, where the equality follows from the
  assumption that~$R$ and~$S$ are probability distributions that are
  consistent.  Write~$P := R \Join_{\mathrm{H}} S$
  and~$Q := R \Join_{\mathrm{P}} S$.  By~\eqref{eqn:coincides} we
  have~$Q \equiv R \Join S$, so~$Q$ witnesses the consistency of~$R$
  and~$S$ by Lemma~\ref{lem:characterization}. Moreover, by
  design,~$Q$ is a probability distribution, and so are~$R$ and~$S$ by
  assumption, so~$Q[Z] = R[Z] = S[Z] = U$. Since~$P$ is also a
  feasible solution of~\eqref{eqn:linearprogram}, also~$P$ is a
  probability distribution that witnesses the consistency of~$R$
  and~$S$, so~$P[Z] = R[Z] = S[Z] = U$. The conclusion of these is
  that
  \begin{align}
    & P[Z] = Q[Z] = U, \label{eqn:marginalurs1} \\
    & P[X] = Q[X] = R, \label{eqn:marginalurs2} \\
    & P[Y] = Q[Y] = S. \label{eqn:marginalurs3}
  \end{align}
  In particular~$\Entropy_P(XY) \geq
  \Entropy_Q(XY)$ since~$P$
  maximizes~\eqref{eqn:entropy}. We show
  that~$\Entropy_P(XY) \leq \Entropy_Q(XY)$, from which it will
  follow that~$P = Q$ since we argued already that the maximum
  of~\eqref{eqn:entropy} is unique.

  Let~$\KL(P||Q)$ denote the Kullback-Leibler divergence (see
  Section~2.3 in~\cite{CoverThomas}) between two probability
  distributions~$P(X)$ and~$Q(X)$ over the same set of attributes~$X$,
  which is defined as
\begin{equation}
  \KL(P||Q) := \sum_{x \in P'} P(x)\log(P(x)/Q(x)),
\end{equation}
with the conventions that~$0\log(0/q)=0$ and~$p\log(p/0) = \infty$.
The Information Inequality (Theorem~2.6.3 in~\cite{CoverThomas})
states that~$\KL(P||Q) \geq 0$. Therefore
  \begin{align}
    \Entropy_P(XY) & = -\sum_{t \in P'} P(t) \log(P(t)) \leq -\sum_{t \in P'} P(t) \log(Q(t)). \label{eqn:step1}
  \end{align}
  Using~\eqref{eqn:identity}, the right-hand side of~\eqref{eqn:step1}
  equals
  \begin{equation}
    \sum_{t \in P'} P(t) \log(U(t[Z]))) -\sum_{t \in P'}
    P(t)\log(R(t[X])) - \sum_{t \in P'} P(t) \log(S(t[Y]))
    \label{eqn:threeway}
  \end{equation}
  Splitting the set of tuples $t$ in $P'$ by $t[Z]$, the first term
  in~\eqref{eqn:threeway} rewrites into
  \begin{align}
    \sum_{u \in U'} \sum_{t \in P' : \atop t[Z]=u} P(t) \log(U(u)) = \sum_{u \in U'} \log(U(u)) P(u)
    = \sum_{u \in U'} \log(U(u)) Q(u)
    \label{eqn:turu1}
  \end{align}
  where the first equality follows
  from~\eqref{eqn:marginal}, and the second follows
  from~\eqref{eqn:marginalurs1}. Exactly the same argument for the second and third terms
  in~\eqref{eqn:threeway}, and applying~\eqref{eqn:marginal} to $Q(u)$, $Q(r)$, and $Q(s)$,
  rewrites~\eqref{eqn:threeway} into
    \begin{equation}
    \sum_{t \in Q'} Q(t)
    \log(U(t[Z]))) - \sum_{t \in Q'} Q(t)\log(R(t[X])) - \sum_{t \in Q'} Q(t) \log(S(t[Y]))
    \label{eqn:threeway2}
    \end{equation}
  and therefore, by~\eqref{eqn:identity}, into
  \begin{equation}
   -\sum_{t \in P'} Q(t) \log(Q(t)) = \Entropy_Q(XY). \label{eqn:stepfinal}
  \end{equation}
  Combining~\eqref{eqn:step1},~\eqref{eqn:threeway},~\eqref{eqn:threeway2},
  and~\eqref{eqn:stepfinal} we get $\Entropy_P(XY) \leq
  \Entropy_Q(XY)$ as was to be shown.
\end{proof}

Next, we show that~$\Join_{\mathrm{P}}$ also maximizes conditional
entropy; it follows that~$R \Join_{\mathrm{CH}} S = R \Join_{\mathrm{H}} S$.

\begin{lemma} \label{lem:maximizesconditionalentropy}  If~$R$ and~$S$ are two consistent probability distributions,  then~
$R \Join_{\mathrm{P}} S = R \Join_{\mathrm{CH}} S$.
\end{lemma}

\begin{proof}
  Let~$X$ and~$Y$ be the sets of attributes of~$R$ and~$S$,
  write~$Z = X \cap Y$, and assume that~$R$ and~$S$ are consistent.
  Let~$U := R[Z] = S[Z]$, where the equality follows from the
  assumption that~$R$ and~$S$ are probability distributions that are
  consistent.  Write~$P := R \Join_{\mathrm{CH}} S$
  and~$Q := R \Join_{\mathrm{P}} S$.  By~\eqref{eqn:coincides} we
  have~$Q \equiv R \Join S$, so~$Q$ witnesses the consistency of~$R$
  and~$S$ by Lemma~\ref{lem:characterization}. Moreover, by
  design,~$Q$ is a probability distribution, and so are~$R$ and~$S$ by
  assumption, hence~$Q[Z] = R[Z] = S[Z] = U$. Since~$P$ is also a
  feasible solution of~\eqref{eqn:linearprogram}, also~$P$ is a
  probability distribution that witnesses the consistency of~$R$
  and~$S$, hence~$P[Z] = R[Z] = S[Z] = U$. The conclusion of these is
  that~$P[Z] = Q[Z] = U$ and both~$P$ and~$Q$ are feasible solutions
  of~\eqref{eqn:linearprogram}. In
  particular, $\Entropy_P(XY|Z) \geq \Entropy_Q(XY|Z)$ since~$P$
  maximizes~\eqref{eqn:conditionalentropy}. We show
  that~$\Entropy_P(XY|Z) \leq \Entropy_Q(XY|Z)$, from which it will
  follow that~$P = Q$ since we argued already that the maximum
  of~\eqref{eqn:entropy} is unique.

  We introduce a piece of notation. Let~$X_0 := X\setminus Z$
  and~$Y_0 := Y\setminus Z$. For each~$Z$-tuple~$u \in U'$ we
  write~$P_u$ and~$Q_u$ to denote the probability distributions
  over~$X_0Y_0$ defined by~$P_z(w) := P(wu)/P(u)$
  and~$Q_z(w) := Q(wu)/Q(u)$ for every~$X_0Y_0$-tuple~$w$. Using the
  obvious fact that if~$D(X)$ is a probability distribution over~$X$
  and~$Z \subseteq Y \subseteq X$
  then~$\Entropy_D(Z) = \Entropy_{D(Y)}(Z)$, we have
  \begin{align}
  & \Entropy_P(XY|Z) = \textstyle{\sum_{u \in U'} U(u) \Entropy_{P_u(X_0Y_0)}(X_0Y_0)}, \\
  & \Entropy_Q(XY|Z) = \textstyle{\sum_{u \in U'} U(u) \Entropy_{Q_u(X_0Y_0)}(X_0Y_0)}.
  \end{align}
  Thus, to prove that~$\Entropy_P(XY|Z) \leq \Entropy_Q(XY|Z)$ it
  suffices to show
  that~$\Entropy_{P_u}(X_0Y_0) \leq \Entropy_{Q_u}(X_0Y_0)$ for
  each~$u \in U'$. Now note that for every~$X_0$-tuple~$r_0$ and
  every~$Y_0$-tuple~$s_0$, and every~$u \in U'$, we have
  \begin{align}
    Q_u(r_0s_0) & = Q(r_0s_0u)/Q(u) = R(r_0u)S(s_0u)/U(u)^2 = R_u(r_0) S_u(s_0), \label{eqn:firstthis}
  \end{align}
  where the first follows from the definition of~$Q_u$, the second
  from~\eqref{eqn:identity} and~$Q(u)=U(u)$, and the third follows
  from setting~$R_u(r_0) := R(r_0u)/R(u)$
  and~$S_u(s_0) := S(s_0u)/S(u)$ and the fact that~$R(u)=S(u)=U(u)$.
  Now recall that~$P[X] = Q[X] = R$, so~$P_u[X_0] = Q_u[X_0] = R_u$
  for every~$u \in U'$, and also~$P[Y] = Q[Y] = S$,
  so~$P_u[Y_0] = Q_u[Y_0] = S_u$ for every~$u \in U'$.  The conclusion
  is that the marginals of~$P_u$ and~$Q_u$ agree, and those of~$Q_u$
  are independent by~\eqref{eqn:firstthis}.  It follows that
  \begin{align}
    \Entropy_{P_u}(X_0Y_0) & \leq \Entropy_{P_u(X_0)}(X_0) + \Entropy_{P_u(Y_0)}(Y_0) \\
    & = \Entropy_{Q_u(X_0)}(X_0) + \Entropy_{Q_u(Y_0)}(Y_0) \\
    & = \Entropy_{Q_u}(X_0Y_0)
  \end{align}
  where the first follows
  from~$\KL(P_u(X_0Y_0)||P_u(X_0)P_u(Y_0)) \geq 0$ by the Information
  Inequality (Theorem~2.6.3 in~\cite{CoverThomas}), the second follows
  from equal marginals, and the third follows from the fact
  that~$\KL(Q_u(X_0Y_0)||Q_u(X_0)Q_u(Y_0)) \geq 0$ holds with equality
  if and only if the marginals~$Q_u(X_0)$ and~$Q_u(Y_0)$ are
  independent (see again Theorem~2.6.3 in~\cite{CoverThomas}), which
  we argued is the case for~$Q_u$.
\end{proof}

The following result is an  immediate consequence of Lemmas~\ref{lem:maximizesentropy}
and~\ref{lem:maximizesconditionalentropy}.

\begin{corollary}
  Let~$R(X)$ and~$S(Y)$ be probability distributions, and
  let~$Z = X \cap Y$. If~$R$ and~$S$ are consistent, then the
  probability distributions~$P$ and~$Q$ among those
  in~$\mathrm{W}(R,S)$ that maximize entropy~$\Entropy_P(XY)$ and
  conditional entropy~$\Entropy_Q(XY|Z)$ are equal. Moreover, if~$R$
  and~$S$ are rational, then~$P$ and~$Q$ are rational.
\end{corollary}

Summarizing, we have proved that whenever~$R$ and~$S$ are consistent
probability distributions we
have~$R \Join S \equiv R \Join_{\mathrm{P}} S$ and
$R \Join_{\mathrm{P}} S = R \Join_{\mathrm{H}} S = R
\Join_{\mathrm{CH}} S$, which we view as evidence that our definition
of $\Join$ is well motivated.

\paragraph{Lossless Join Decompositions}

We provide further justification for the definition of the join of
two~$K$-relations by showing that a decomposed~$K$-relation can be
reconstructed (up to equivalence) by joining its decomposed parts,
under the same hypothesis that makes it possible to reconstruct an
ordinary relation by joining its decomposed parts. This justification
for the~$\Join$ operation is valid for an arbitrary positive semiring $K$.

Let~$U$ be a set of attributes, let~$P$ be an ordinary relation
over~$U$, and let~$V,W$ and~$X,Y$ be pairs of subsets of~$U$. We say
that~$P$ \emph{satisfies the functional dependency}~$V \rightarrow W$
if whenever two tuples in~$P$ agree on all attributes in~$V$, then
they also agree on all attributes in~$W$.  The \emph{decomposition
  of~$P$ along~$X$ and~$Y$} consists of the projections~$R = P[X]$
and~$S= P[Y]$ of~$P$ on the sets~$X$ and~$Y$, respectively.  Such a
decomposition is said to be a \emph{lossless-join} decomposition
if~$P=R \Join S$, that is, the relation~$P$ can be reconstructed by
joining the parts~$R = P[X]$ and~$S= P[Y]$ of the decomposition.

The following lemma gives a sufficient condition for a decomposition
to be a lossless join one. Even though this lemma is standard
textbook material, we include a proof for completeness and comparison
with what is to follow.

\begin{lemma} \label{prop:lossless} Let~$P$ be an ordinary
  relation that is decomposed along~$X$ and~$Y$. If~$P$ satisfies the
  functional dependency~$X\cap Y \rightarrow X\setminus Y$ or~$P$
  satisfies the functional
  dependency~$X\cap Y \rightarrow Y \setminus X$, then this
  decomposition is lossless-join.
\end{lemma}

\begin{proof}
  From the definitions, it follows that we always
  have~$P \subseteq R\Join S$.  Assume that~$P$ satisfies the
  functional dependency~$X\cap Y \rightarrow X\setminus Y$ (the other
  case is proved using a symmetric argument).  We will show
  that~$R\Join S \subseteq P$. Let~$t$ be a tuple in~$R\Join S$. It
  follows that~$t[X] \in R = P[X]$ and~$t[Y] \in S = P[Y]$. Therefore,
  there are tuples~$t_1$ and~$t_2$ in~$P$, such that~$t[X]=t_1[X]$
  and~$t[Y] = t_2[Y]$. Since~$P$ satisfies the functional
  dependency~$X\cap Y \rightarrow X\setminus Y$ and
  since~$t_1[X\cap Y] = t[X\cap Y] = t_2[X \cap Y]$, we must have
  that~$t_1[X\setminus Y]= t_2[X\setminus Y]$. Since~$t[X] = t_1[X]$
  and~$t[Y] = t_2[Y]$, it follows that~$t = t_2$, hence~$t \in P$;
  this completes the proof that~$R\Join S \subseteq P$.
 \end{proof}

 It is easy to see that there are lossless-join decompositions of
 relations that satisfy neither the functional
 dependencies~$X\cap Y \rightarrow X\setminus Y$ nor the functional
 dependency~$X\cap Y \rightarrow Y \setminus X$. Thus,
 Lemma~\ref{prop:lossless} is a sufficient, but not necessary,
 condition for a decomposition to be a lossless join one. The
 condition, however, is necessary and sufficient for relations over a
 schema that satisfy a set functional dependencies. To make this
 statement precise, we recall a basic definition from relational
 databases.
 Let~$U$ be a set of attributes, let~$F$ be a set of functional
 dependencies between subsets of~$U$, and let~$V\rightarrow W$ be a
 functional dependency.  We say that~$F$ \emph{logically
   implies}~$V\rightarrow W$, denoted~$F\models V\rightarrow W$ if
 whenever a relation~$R$ satisfies every functional dependency in~$F$,
 then~$R$ also satisfies~$V\rightarrow W$.  The following is a well
 known result in relational database theory (see,~e.g.,~Theorem~7.5
 in~\cite{DBLP:books/cs/Ullman88}).

 \begin{theorem} \label{lossless-join:thm} Let~$U$ be a set of
   attributes, let~$F$ be a set of functional dependencies between
   subsets of~$U$, and let~$X$ and~$Y$ are two subsets of~$U$. Then
   the following statements are equivalent:
\begin{enumerate}
\item[(a)] $F\models X\cap Y \rightarrow X\setminus Y$
  or~$F\models X\cap Y\rightarrow Y\setminus X$.
\item[(b)] For every relation~$R$ over~$U$ that satisfies every
  functional dependency in~$F$, it holds that if~$R$ is decomposed
  along~$X$ and~$Y$, then this decomposition is a lossless-join one.
    \end{enumerate}
    \end{theorem}

    Our next result tells that Lemma~\ref{prop:lossless} extends
    to decompositions of~$K$-relations, where~$K$ is a positive
    semiring. We first need to extend the notions appropriately.
    If~$P$ is a~$K$-relation, then the \emph{decomposition of~$P$
      along~$X$ and~$Y$} consists of the marginals~$R=P[X]$
    and~$S=P[Y]$. We say that the decomposition is
    \emph{lossless-join} if~$P \equiv R\Join S$, where~$\Join$ is the
    join operation on~$K$-relations.

    \begin{lemma} \label{prop:lossless-gen} Let~$P$ be
      a~$K$-relation that is decomposed along~$X$ and~$Y$. If the
      support~$P'$ of~$P$ satisfies the functional
      dependency~$X\cap Y \rightarrow X\setminus Y$ or~$P'$ satisfies
      the functional dependency~$X\cap Y \rightarrow Y \setminus X$,
      then this decomposition is lossless-join.
    \end{lemma}

    \begin{proof} For concreteness, let us assume that the attributes
      of~$P$ are~$ABC$ and that~$P$ is decomposed along~$X = AB$
      and~$Y = BC$. The proof remains the same in the general case and
      with only notational changes.

      We will show that~$P \equiv R \Join S$, where~$R=P[X]$
      and~$S = P[Y]$ . In fact, we will show
      that~$R\Join S = c^*_{S,X\cap Y} P$. By Part~2 of
      Lemma~\ref{lem:easyfacts1} we have~$R'=P'[X]$ and~$S'=P'[Y]$. In
      addition, since~$P'$ satisfies the functional
      dependency~$X \cap Y \rightarrow X \setminus Y$ or the
      functional dependency~$X \cap Y \rightarrow Y \setminus X$ we
      have~$P'= R' \Join S'= (R \Join S)'$, where the first equality
      follows from Lemma~\ref{prop:lossless} and the second from
      Lemma~\ref{lem:supportsjoin}.  Let~$(a,b,c)$ be a tuple
      in~$(R\Join S)'$, so in particular~$(a,b,c) \in P'$. By the
      definition of~$R\Join S$, we have
      that~$(R \Join S)(a,b,c) = R(a,b)S(b,c)c_S(b)$. We now examine
      the quantities~$R(a,b)$ and~$S(b,c)$ separately, and for that we
      distinguish by cases.
      \bigskip

      \noindent\textit{Case 1}: The ordinary relation~$P'$ satisfies the
      functional dependency~$X \cap Y \rightarrow X \setminus Y$,
      which, in this case, amounts to~$B\rightarrow A$. For $R(a,b)$
      we have
  \begin{equation}
    R(a,b) = \sum_{c':(a,b,c') \in P'} P(a,b,c') = \sum_{a',c':(a',b,c') \in P'} P(a',b,c') =  P(b),
    \label{eqn:hhg11}
\end{equation}
  where the first follows from~$R = P[X]$ and~\eqref{eqn:marginal},
  the second follows from the fact that, since~$P'$ satisfies the
  functional dependency~$B \rightarrow A$, we must have
  that~$(a',b,c') \in P'$ implies~$a'=a$, and the third follows
  from~\eqref{eqn:marginal}. For~$S(b,c)$ we have
  \begin{equation}
  S(b,c) = \sum_{a':(a',b,c)\in P'}P(a',b,c) = P(a,b,c), \label{eqn:hhg12}
  \end{equation}
  where the first follows from~$S = P[Y]$ and~\eqref{eqn:marginal},
  and the second follows from the fact that, since~$P'$ satisfies the
  functional dependency~$B \rightarrow A$ and~$(a,b,c) \in P'$, we
  must have that~$(a',b,c) \in P'$ holds if and only if~$a' = a$.
  \bigskip

  \noindent \textit{Case 2}: The ordinary relation~$P'$ satisfies the functional
  dependency~$X \cap Y \rightarrow Y \setminus X$, which, in this
  case, amounts to~$B\rightarrow C$.  Using a similar analysis as in
  the previous case, for~$S(b,c)$ we have that
\begin{equation}
S(b,c) = \sum_{a':(a',b,c) \in P'} P(a',b,c) = \sum_{a',c':(a',b,c')} P(a',b,c') = P(b),
\label{eqn:hhg21}
\end{equation}
where the first follows from~$S = P[Y]$ and~\eqref{eqn:marginal}, the
second follows from the fact that, since~$P'$ satisfies the functional
dependency~$B \rightarrow C$, we must have that~$(a',b,c') \in P'$
implies~$c'=c$, and the third follows from~\eqref{eqn:marginal}. For $R(a,b)$ we have
\begin{equation}
R(a,b) = \sum_{c':(a,b,c') \in P'} P(a,b,c') = P(a,b,c)
\label{eqn:hhg22}
\end{equation}
where the first follows from~$R = P[X]$ and~\eqref{eqn:marginal}, and
the second follows from the fact that, since~$P'$ satisfies the
functional dependency~$B \rightarrow C$ and~$(a,b,c) \in P'$ we must
have that~$(a,b,c') \in P'$ holds if and only if~$c'=c$.

In both cases, it follows that
\begin{eqnarray}
  (R \Join S)(a,b,c)  = R(a,b)S(b,c)c_S(b)
  = P(b) P(a,b,c) c_S(b) = c^*_{P,B} P(a,b,c),
\end{eqnarray}
where the first follows from~\eqref{eqn:joinasymmetric}, the second
follows from~\eqref{eqn:hhg11} and~\eqref{eqn:hhg12} in one case, and
from~\eqref{eqn:hhg21} and~\eqref{eqn:hhg22} in the other, and the
last follows from~$c^*_{P,B} = c_S(b) P(b)$ by \eqref{eqn:cr}
since~$b \in P[B]' = P'[B]$ given that~$(a,b,c) \in P'$ and Part~2 of
Lemma~\ref{lem:easyfacts1}. This proves
that~$R \Join S = c^*_{P,B} P$, which was to be shown.
\end{proof}

The last result in this section asserts that the preceding Theorem \ref{lossless-join:thm} extends to decompositions of $K$-relations.

\begin{proposition} \label{lossless-join-K:thm} Let $K$ be a positive semiring, let $U$ be a set of
  attributes, let $F$ be a set of functional dependencies between
  subsets of $U$, and let $X$ and $Y$ be two subsets of $U$.  The
  following statements are equivalent:
 \begin{enumerate}  \itemsep=0pt
\item[(a)] $F\models X\cap Y \rightarrow X\setminus Y$ or $F\models X\cap Y\rightarrow Y\setminus X$.
\item[(b)] For every $K$-relation $R$ over $U$ whose support $R'$ satisfies every functional dependency in $F$, it holds that if $R$ is decomposed along $X$ and $Y$, then this decomposition is a lossless-join one.
    \end{enumerate}
    \end{proposition}
\begin{proof}

  First, assume that~$F\models X\cap Y \rightarrow X\setminus Y$
  or~$F\models X\cap Y\rightarrow Y\setminus X$. Let~$R$ be a be
  a~$K$-relation whose support~$R'$ satisfies every functional
  dependency in~$F$. Then~$R'$ satisfies the functional
  dependency~$X \cap Y \rightarrow X \setminus Y$ or~$R'$ satisfies
  the functional dependency~$X \cap Y \rightarrow Y \setminus X$. By
  Lemma~\ref{prop:lossless-gen}, the decomposition of~$R$ along~$X$
  and~$Y$ is lossless join.

  Next, assume that for every~$K$-relation~$R$ over~$U$ whose
  support~$R'$ satisfies every functional dependency in~$F$, it holds
  that if~$R$ is decomposed along~$X$ and~$Y$, then this decomposition
  is a lossless-join one.  Let~$P$ be an arbitrary ordinary relation
  over~$U$ that satisfies every functional dependency in~$F$. We will
  show that if~$P$ is decomposed along~$X$ and~$Y$, then the
  decomposition is a lossless join one, hence, by Theorem
  \ref{lossless-join:thm}, we have
  that~$F\models X\cap Y \rightarrow X\setminus Y$
  or~$F\models X\cap Y\rightarrow Y\setminus X$.  Turn~$P$ into
  a~$K$-relation~$R$ whose support is~$P$ and where all tuples in the
  support have value~$1$ in~$K$.  In other words, consider
  the~$K$-relation~$R$ such that for every tuple~$t$, we have
  that~$R(t)=1$ if~$t \in P$, and~$R(t) = 0$ if~$t\not \in P$. By
  hypothesis, the decomposition of~$R$ along~$X$ and~$Y$ is a lossless
  join one. Therefore,~$R \equiv R[X] \Join R[Y]$, as~$K$-relations.
  By Lemma~\ref{lem:supportsjoin} and Part~2 of
  Lemma~\ref{lem:easyfacts2}, we have that~$R' = R[X]' \Join R[Y]'$,
  i.e., the support of the join is the join of the supports as
  ordinary relations. From the definition of~$R$, we have that~$R'=P$;
  moreover, from Part~2 of Lemma \ref{lem:easyfacts1}, we have
  that~$R[X]'= R'[X]$ and~$R[Y]'= R'[Y]$. Since~$R'[X] = P[X]$
  and~$R'[Y]= P[Y]$, we conclude that~$P = P[X]\Join P[Y]$, thus the
  decomposition of~$P$ along~$X$ and~$Y$ is a lossless join one, which
  was to be shown.
\end{proof}

\section{Consistency of Three or More $K$-Relations} \label{sec:cons-more-than-two}

While the definition of consistency of two~$K$-relations has a
straightforward generalization to the case of three or
more~$K$-relations, not all the related concepts will go through: the
join of three or more~$K$-relations will be particularly
problematic. We start with the definitions.

Let~$K$ be a positive semiring and let~$R_1(X_1),\ldots,R_m(X_m)$
be~$K$-relations. We say that the collection~$R_1,\ldots,R_m$ is
\emph{globally consistent} if there is a~$K$-relation~$T$
over~$X_1 \cup \cdots \cup X_m$ such that~$R_i \equiv T[X_i]$ for
all~$i \in [m]$.  We say that such a~$K$-relation \emph{witnesses} the
global consistency of~$R_1,\ldots,R_m$.  The equivalence
classes~$[R_1],\ldots,[R_m]$ are called \emph{globally consistent} if
their representatives~$R_1,\ldots,R_m$ are globally consistent. As in
the case of two equivalence classes, it is easy to see using
transitivity of~$\equiv$ that this notion is well-defined in that it
does not depend on the chosen representatives.

We also say that the relations~$R_1,\ldots,R_m$ are \emph{pairwise
  consistent} if for every~$i,j \in [m]$ we have that~$R_i[X_i]$
and~$R_j[X_j]$ are consistent. From the definitions, it follows that
if~$R_1,\ldots,R_m$ are globally consistent, then they are also
pairwise consistent. The converse, however, need not be true, in
general. In fact, the converse fails even for ordinary relations, that
is, for~$\mathbb{B}$-relations, where~$\mathbb{B}$ is the Boolean
semiring.
For example, it is easy to see that the ordinary relations~$R(AB) =
\{00, 11\}$,~$S(BC) = \{01, 10\}$,~$T(AC)=
\{00,11\}$ are pairwise consistent but not globally consistent.

In the context of relational databases, there has been an extensive
study of global consistency for ordinary relations. We present an
overview of some of the main findings next.

\subsection{Global Consistency in the Boolean Semiring} \label{subsec:cons-boole}

For this section~$K$ is the Boolean semiring~$\mathbb{B}$ and
therefore~$K$-relations are ordinary relations or, simply,
\emph{relations}.  Let~$R_1,\ldots, R_m$ be a collection of
relations. The \emph{relational join} or, simply, the~\emph{join}
of~$R_1,\dots,R_m$ is the relation~$R_1\Join \cdots \Join R_m$
consisting of all~$(X_1\cup \cdots \cup X_m)$-tuples~$t$ such
that~$t[X_i]$ belongs to~$R_i$ for all~$i = 1,\ldots,m$. The following
facts are well known and easy to prove (e.g., see
\cite{DBLP:journals/ipl/HoneymanLY80}):
\begin{itemize} \itemsep=0pt
\item If $T$ is a relation witnessing the global consistency of  $R_1,\ldots,R_m$, then $T \subseteq R_1\Join \cdots \Join R_m$.
\item The collection~$R_1,\ldots,R_m$ is
  globally consistent if and only
  if~$(R_1\Join \cdots \Join R_m)[X_i]= R_i$ for all~$i = 1,\ldots,m$.
\end{itemize}
Consequently, if the collection $R_1,\ldots,R_m$ is globally
consistent, then the join $R_1\Join \cdots \Join R_m$ is the largest
relation witnessing their consistency.


As seen earlier, pairwise consistency is a necessary, but not
sufficient, condition for global consistency. This was exemplified by
three relations~$R,S,T$ with schema~$AB$,~$BC$,~$AC$, respectively. In
contrast, it is not hard to see that if the schema of the three
relations had been~$AB$,~$BC$,~$CD$, then pairwise consistency would
have been a necessary and sufficient condition for the global
consistency of any three~$K$-relations over these schema. This raises
the question whether it is possible to characterize the set of schema
for which pairwise consistency is a necessary and sufficient condition
for global consistency. This question was investigated and answered by
Beeri, Fagin, Maier, and Yannakakis
\cite{BeeriFaginMaierYannakakis1983}. Before describing their results,
we need to introduce a number of notions from hypergraph theory.

\paragraph{Hypergraphs}
A \emph{hypergraph} is a pair $H = (V,E)$, where $V$ is a set of
\emph{vertices} and $E$ is a set of \emph{hyperedges}, each of which
is a non-empty subset of $V$.  Clearly, the undirected graphs without
self-loops are precisely the hypergraphs all the hyperedges of which
are two-element sets of vertices; such hyperedges are called
\emph{edges}.

Let~$H = (V,E)$ be a hypergraph. The \emph{reduction} of~$H$, denoted
by~$R(H)$, is the hypergraph whose set of vertices is~$V$ itself and
whose hyperedges are those hyperedges~$X \in E$ that are not included
in any other hyperedge of~$H$. A hypergraph~$H$ is \emph{reduced}
if~$H=R(H)$.  If~$W \subseteq V$, then the \emph{hypergraph induced
  by~$W$ on~$H$}, denoted by~$H[W]$, is the hypergraph whose set of
vertices is~$W$ and whose hyperedges are the non-empty subsets of the
form~$X \cap W$, where~$X \in E$ is a hyperedge of~$H$; in
symbols,~$H[W] = (W, \{X \cap W: X \in E\}\setminus \{\emptyset\})$.

Every collection~$X_1,\ldots,X_m$ of sets of attributes can be
identified with a hypergraph~$H=(V,E)$,
where~$V = X_1\cup \cdots \cup X_m$ and~$E
=\{X_1,\ldots,X_m\}$. Conversely, every hypergraph~$H = (V,E)$ gives
rise to a collection~$X_1,\ldots,X_m$ of sets of attributes,
where~$X_1,\dots,X_m$ are the hyperedges of~$H$. For this reason, we
can move from collections of sets of attributes to hypergraphs (and
vice versa) in a seamless way.  In what follows, we will consider
several structural properties of hypergraphs that, as shown in
\cite{BeeriFaginMaierYannakakis1983}, give rise to necessary and
sufficient conditions for pairwise consistency to coincide with global
consistency.

\paragraph{Acyclic Hypergraphs}
We begin by defining the notion of an acyclic hypergraph, which
generalizes the notion of an acyclic undirected graph with no
self-loops.  For this, we need to introduce several auxiliary
notions. If~$H$ is a hypergraph and~$u$ and $v$ are vertices of~$H$,
then a~\emph{path from~$u$ to~$v$} is a sequence~$Y_1,\ldots,Y_k$ of
hyperedges of~$H$, for some positive integer~$k$, such that~$u\in Y_1$
and~$v\in Y_k$, and~$Y_i\cap Y_{i+1}\neq \emptyset$, for
all~$i \in [k-1]$.  Using the notion of a path, one defines the
notions of a \emph{connected component} of a hypergraph and of a
\emph{connected} hypergraph in the obvious way.

Let~$H=(V,E)$ be a reduced hypergraph and let~$X$ and~$Y$ be two
distinct hyperedges of~$H$ with~$X\cap Y \not= \emptyset$. We say
that~$X\cap Y$ is an \emph{articulation set} of~$H$ if the number of
connected components of the reduction~$R(H[V\setminus (X\cap Y)])$
of~$H[V\setminus (X\cap Y)]$ is greater than the number of the
connected components of~$H$.  We say that a reduced hypergraph~$H$ is
\emph{acyclic} if for every subset~$W$ of vertices of~$H$,
if~$R(H[W])$ is connected and has more than one hyperedges, then it
has an articulation set.  Finally, a hypergraph~$H$ is \emph{acyclic}
if its reduction~$R(H)$ is acyclic.  Otherwise,~$H$ is
\emph{cyclic}. We say that a collection~$X_1,\ldots,X_m$ of sets of
attributes is \emph{acyclic} if the hypergraph with
hyperedges~$X_1,\ldots,X_m$ is acyclic; otherwise, we say that the
collection~$X_1,\ldots,X_m$ is \emph{cyclic}.

To illustrate these concepts, consider the set~$V=\{A_1,\ldots,A_n\}$
of vertices and the sets~$E_1$ and~$E_2$ of hyperedges,
where~$E_1=\{ \{A_i,A_{i+1}\}: 1\leq i\leq n-1\}$
and~$E_2 = E_1 \cup \{\{A_n,A_1\}\}$. It is not hard to verify that
the hypergraph~$H_1=(V,E_1)$ is acyclic, whereas the
hypergraph~$H_2=(V,E_2)$ is cyclic as soon as~$n \geq 3$.

\paragraph{Conformal and Chordal Hypergraphs} The \emph{primal} graph
of a hypergraph~$H = (V,E)$ is the undirected graph that has~$V$ as
its set of vertices and has an edge between any two distinct vertices
that appear together in at least one hyperedge of~$H$. A
hypergraph~$H$ is \emph{conformal} if the set of vertices of every
clique (i.e., complete subgraph) of the primal graph of~$H$ is
contained in some hyperedge of~$H$.  For example, both
hypergraphs~$H_1$ and~$H_2$ above are conformal, whereas the
hypergraph~$H_3=(V,E_3)$ with~$V=\{A_1,\ldots,A_n\}$
and~$E_3= \{ V\setminus \{A_i\}: 1\leq i\leq n\}$ is not conformal as
long as~$n \geq 3$.

A hypergraph~$H$ is \emph{chordal} if its primal graph is chordal,
that is, if every cycle of length at least four of the primal graph
of~$H$ has a chord. For example, the hypergraphs~$H_1$ and~$H_3$ above
are chordal, whereas the hypergraph~$H_2$ is not chordal
for~$n \geq 4$. Observe that, in case~$n = 3$, the hypergraph~$H_2$ is
chordal but not conformal.

\paragraph{Running Intersection Property} We say that a hypergraph~$H$
has the \emph{running intersection property} if there is a
listing~$X_1,\ldots,X_m$ of all hyperedges of~$H$ such that for
every~$i \in [m]$ with~$i \geq 2$, there exists a~$j < i$ such
that~$X_i \cap (X_1 \cup \cdots \cup X_{i-1}) \subseteq X_j$.

For example, the hypergraph~$H_1$ has the running intersection
property with the listing
being~$\{A_1,A_2\}, \ldots, \{A_{n-1},A_n\}$, whereas the
hypergraphs~$H_2$ and~$H_3$ do not have the running intersection
property as long as~$n \geq 3$.

\paragraph{Join Trees} A \emph{join tree} for a hypergraph~$H$ is an
undirected tree~$T$ with the set~$E$ of the hyperedges of~$H$ as its
vertices and such that for every vertex~$v$ of~$H$, the set of
vertices of~$T$ containing~$v$ forms a subtree of~$T$, i.e., if~$v$
belongs to two vertices~$X_i$ and~$X_j$ of~$T$, then~$v$ belongs to
every vertex of~$T$ in the unique path from~$X_i$ to~$X_j$ in~$T$.

For example, the hypergraph~$H_1$ has a join tree (in fact, the join
tree is a path) with edges of the
form~$\{\{A_i,A_{i+1}\},\{A_{i+1},A_{i+2}\}\}$ for~$i \in [n-2]$,
whereas the hypergraphs~$H_2$ and~$H_3$ do not have a join tree
for~$n \geq 3$.

\paragraph{Graham's Algorithm} Consider the following iterative
algorithm on hypergraphs: given a hypergraph~$H=(V,E)$, apply the
following two operations repeatedly until neither of the two
operations can be applied:
\begin{enumerate} \itemsep=0pt
\item If~$v$ is a vertex that appears in only one hyperedge~$X_i$
  of~$H$, then delete~$v$ from~$X_i$.
\item If there are two hyperedges~$X_i$ and~$X_j$ such
  that~$i\not = j$ and~$X_i\subseteq X_j$, then delete~$X_i$ from~$E$.
\end{enumerate}
It can be shown that this algorithm has the Church-Rosser property, that
is, it produces the same hypergraph independently of the order in
which the above two operations are applied.

We say that \emph{Graham's algorithm succeeds} on a
hypergraph~$H=(V,E)$ if the algorithm, given~$H$ as input, returns the
empty hypergraph~$H= (V, \emptyset)$ as output. Otherwise, we say that
\emph{Graham's algorithm fails} on~$H$. We also say that~$H$ is
\emph{accepted} by Graham's algorithm if the algorithm succeeds
on~$H$; otherwise we say that it is \emph{rejected}.

For example, Graham's algorithm succeeds on the hypergraph~$H_1$,
whereas it fails on the hypergraphs~$H_2$ and~$H_3$ as long
as~$n \geq 3$ (in fact, it returns~$H_2$ and~$H_3$, respectively).

This algorithm was designed by Graham \cite{Graham79}.  A similar
algorithm was designed by Yu and Ozsoyoglu \cite{yu1979algorithm};
these two algorithms are often referred to as the GYO Algorithm (see
\cite{DBLP:books/aw/AbiteboulHV95}).

\paragraph{Local-to-Global Consistency Property} The notions
introduced so far can be thought of as ``syntactic" or ``structural"
properties that some hypergraphs possess and others do not, because
their definitions involve only the vertices and the hyperedges of the
hypergraph at hand.  In contrast, the next notion is ``semantic", in
the sense that its definition also involves relations whose sets of
attributes are the hyperedges of the hypergraph at hand.

Let~$H$ be a hypergraph and let~$X_1,\dots,X_m$ be a listing of all
hyperedges of~$H$. We say that~$H$ has the \emph{\ltgc~for ordinary
  relations} if every pairwise consistent
collection~$R_1(X_1),\ldots,R_m(X_m)$ of relations of
schema~$X_1,\ldots,X_m$ is globally consistent.

For example, it can be shown that the hypergraph~$H_1$ has the
\ltgc~for ordinary relations, whereas the hypergraphs~$H_2$ and~$H_3$
do not have this property as long as~$n \geq 3$.
\bigskip

We are now ready to state the main result in Beeri, Fagin, Maier, Yannakakis \cite{BeeriFaginMaierYannakakis1983}.

\begin{theorem} [Theorem 3.4 in \cite{BeeriFaginMaierYannakakis1983}] \label{thm:BFMY}
Let~$H$ be a hypergraph. The following statements are equivalent:
\begin{enumerate} \itemsep=0pt
\item[(a)] $H$ is an acyclic hypergraph.
\item[(b)] $H$ is a conformal and chordal hypergraph.
\item[(c)] $H$ has the running intersection property.
\item[(d)] $H$ has a join tree.
\item[(e)] $H$ is accepted by Graham's algorithm.
\item[(f)] $H$ has the \ltgc~for ordinary relations.
\end{enumerate}
\end{theorem}

As an illustration of Theorem \ref{thm:BFMY}, let us return to the
hypergraphs $H_1$, $H_2$, $H_3$ encountered earlier. Hypergraph $H_1$
has all six properties in Theorem \ref{thm:BFMY}, whereas hypergraphs
$H_2$ and $H_3$ have none of these properties when $n \geq 3$.

\subsection{Global Consistency in Arbitrary Positive
  Semirings} \label{subsec:cons-pos} Let~$K$ be an arbitrary, but
fixed, positive semiring.  In this section, we investigate some
aspects of global consistency for collections of~$K$-relations.

As discussed earlier, if~$R_1,\ldots,R_m$ is a globally consistent
collection of ordinary relations, then the
join~$R_1\Join \cdots \Join R_m$ witnesses the global consistency
of~$R_1,\ldots,R_m$ (and, in fact, is the largest such witness). At
first, one may expect that a similar result may hold for globally
consistent collections~$R_1,\ldots,R_m$ of~$K$-relations. It turns
out, however, that the concept of the join of three or
more~$K$-relations is problematic, even for the case in which~$K$ is
the bag semiring~${\mathbb N}$ of non-negative integers. Note that,
using the join of two relations, the join of
three~$K$-relations~$R,S,T$ could be defined as
either~$R \Join (S \Join T)$ or as~$(R\Join S) \Join T$. The join of
ordinary relations is associative, hence these two expressions
coincide for ordinary relations. In contrast, there are bags~$R, S, T$
that are globally consistent and such
that
\begin{equation}
R \Join (S \Join T) \not \equiv (R \Join S) \Join T,
\end{equation}
and neither~$R \Join (S \Join T)$ nor~$(R \Join S) \Join T$ witnesses
the global consistency of~$R,S,T$.

\medskip

\begin{example} \label{exam:triangle} Let~$W(ABC)$ be the bag given by
  its table of multiplicities below, along with its three marginals
  $R(AB), S(BC), T(AC)$:
  \begin{center}
  \begin{tabular}{ccccccccccccccccccc}
  $W(ABC)$ \# & \;\;\;\; & $R(AB)$ \# & \;\;\;\; & $S(BC)$ \# & \;\;\;\; & $T(AC)$ \# \\
  \;\;\; 1 1 2 : 1 & & \;\; 1 1 : 1 & & \;\; 1 2 : 1 & & \;\; 1 2 : 1 \\
  \;\;\; 1 2 3 : 2 & & \;\; 1 2 : 2 & & \;\; 1 4 : 4 & & \;\; 1 3 : 2 \\
  \;\;\; 2 1 4 : 4 & & \;\; 2 1 : 4 & & \;\; 2 2 : 2 & & \;\; 2 1 : 3 \\
  \;\;\; 2 2 2 : 2 & & \;\; 2 2 : 2 & & \;\; 2 3 : 2 & & \;\; 2 2 : 2 \\
  \;\;\; 2 3 1 : 3 & & \;\; 2 3 : 3 & & \;\; 3 1 : 3 & & \;\; 2 4 : 4
  \end{tabular}
  \end{center}
  By construction the collection of three bags~$R,S,T$ is globally
  consistent as witnessed by~$W$. We
  produced~$N_1 := (R \Join S) \Join T$
  and~$N_2 := R \Join (S \Join T)$ by computer, along with their
  marginals on~$AB, BC, AC$. We display two bags~$[N_1]$ and~$[N_2]$
  that are in the equivalence classes of~$N_1$ and~$N_2$,
  respectively, along with the two marginals~$P_1 := [N_1][BC]$
  and~$P_2 := [N_2][AB]$ that suffice to verify the claim that neither~$N_1$
  nor~$N_2$ witness the consistency of~$R,S,T$:
\begin{center}
\begin{tabular}{lllllllllllllllllllllllllllll}
 $[N_1](ABC)$\, \# & \;\;\; & $[N_2](ABC)$\; \# & \;\;\; & $P_1(BC)$\, \# & \;\;\; & $P_2(AB)$\, \# \\
  \;\;\;\;\;\;\, 1 1 2 : 14 & & \;\;\;\;\;\;\; 1 1 2 : 1 & & \;\;\;\; 1 2 : 22 & &  \;\;\;\;  1 1 : 1 \\
  \;\;\;\;\;\;\, 1 2 2 : 7 & & \;\;\;\;\;\;\; 1 2 2 : 5 & & \;\;\;\; 1 4 : 48 & &   \;\;\;\;  1 2 : 10 \\
  \;\;\;\;\;\;\, 1 2 3 : 21 & & \;\;\;\;\;\;\; 1 2 3 : 5 & & \;\;\;\; 2 2 : 35 & &  \;\;\;\;  2 1 : 20 \\
  \;\;\;\;\;\;\, 2 1 2 : 8 & & \;\;\;\;\;\;\; 2 1 2 : 4 & & \;\;\;\; 2 3 : 21 & &   \;\;\;\;  2 2 : 5 \\
  \;\;\;\;\;\;\, 2 1 4 : 48 & & \;\;\;\;\;\;\; 2 1 4 : 16 & & \;\;\;\; 3 1 : 42 & & \;\;\;\;  2 3 : 15 \\
  \;\;\;\;\;\;\, 2 2 2 : 28 & & \;\;\;\;\;\;\; 2 2 2 : 5 & & & & & & & & \\
  \;\;\;\;\;\;\, 2 3 1 : 42 & & \;\;\;\;\;\;\; 2 3 1 : 15 & & & & & & & &
\end{tabular}
\end{center}
The ratio test shows that~$[N_1]$ and~$[N_2]$ are not equivalent
($14/1 \not= 7/5$), so~$\Join$ is not associative, not even up to
equivalence.
The ratio test applied to the bags~$P_1(BC)$ and~$S(BC)$ shows that
they are not equivalent~($22/1 \not= 48/4$), and the ratio test
applied to the bags~$P_2(AB)$ and~$R(AB)$ shows that they are not
equivalent~($1/1 \not= 10/2$). Thus, neither~$N_1$ nor~$N_2$ witness
the consistency of~$R,S,T$.
\end{example}

Observe that the collection $AB$, $BC$, $AC$ of the sets of attributes
of the relations in Example \ref{exam:triangle} is cyclic. It turns
out that this is no accident.
Indeed, we show next that if a collection~$X_1,\ldots,X_m$ of sets of
attributes is acyclic and if~$R_1(X_1),\ldots,R_m(X_m)$ is a globally
consistent collection of~$K$-relations of schema~$X_1,\ldots,X_m$,
then a witness of their global consistency can always be built
iteratively through joins of two~$K$-relations. In fact, we show
something stronger, namely, that it suffices for~$R_1,\ldots,R_m$ to
be pairwise consistent~$K$-relations. For stating this lemma we need
the following definition. The~\emph{iterated left-join} of
the~$K$-relations~$R_1,\ldots,R_m$ is the~$K$-relation
\begin{equation}
  ((\cdots(R_1 \Join R_2) \Join \cdots \Join R_{m-2}) \Join
  R_{m-1}) \Join R_m,  \label{eqn:iteratedjoin}
\end{equation}
i.e., the sequential join of~$R_1,\ldots,R_m$ with the join operations
associated to the left. More formally, the iterated left-join
of~$R_1,\ldots,R_m$ is defined by induction on~$m$. For~$m = 1$ it
is~$R_1$, and for~$m \geq 2$ it is~$R \Join R_m$ where~$R$ is the
iterated left-join of~$R_1,\ldots,R_{m-1}$.

\begin{lemma} \label{lem:newlemma} Let~$X_1,\ldots,X_m$ be an acyclic
  collection of sets of attributes. There exists a
  permutation~$\pi : [m] \rightarrow [m]$ such that
  if~$R_1(X_1),\ldots,R_m(X_m)$ are pairwise consistent~$K$-relations
  of schema~$X_1,\ldots,X_m$, then they are globally consistent, and
  the iterated left-join of~$R_{\pi(1)},\ldots,R_{\pi(m)}$ witnesses
  their global consistency.  In particular, if they are globally
  consistent, then the iterated left-join of some permutation of them
  witnesses their global consistency.
\end{lemma}

\begin{proof}
  Assume that~$X_1,\ldots,X_m$ is an acyclic collection of sets of
  attributes.  By Theorem \ref{thm:BFMY}, this collection has the
  running intersection property, hence there exists a
  permutation~$\pi : [m] \rightarrow [m]$ such that for
  every~$i \in [m]$ with~$i\geq 2$, there exists~$j \in [m]$ such
  that~$\pi(j) < \pi(i)$
  and~$X_{\pi(i)} \cap (X_{\pi(1)} \cup \cdots \cup X_{\pi(i-1)})
  \subseteq X_{\pi(j)}$.  By renaming the sets, we may assume
  that~$\pi$ is the identity, so for every~$i \in [m]$
  with~$i \geq 2$, there is a~$j \in [i-1]$ such
  that~$X_i \cap (X_1 \cup \cdots \cup X_{i-1}) \subseteq X_j$. Fix a
  collection of~$K$-relations~$R_1,\ldots,R_m$ for~$X_1,\ldots,X_m$
  and assume that they are pairwise consistent. For each~$i \in [m]$,
  let~$T_i := ((R_{1} \Join \cdots \Join R_{i-2}) \Join
  R_{i-1}) \Join R_i$ with the joins associated to the left. We show,
  by induction on~$i = 1,\ldots,m$, that~$T_i$ is a~$K$-relation
  over~$X_1 \cup \cdots \cup X_i$ that witnesses the consistency
  of~$R_1,\ldots,R_i$.

  For~$i = 1$ the claim is obvious since~$T_1 = R_1$. Assume then
  that~$i \geq 2$ and that the claim is true for smaller indices.
  Let~$X := X_1 \cup \cdots \cup X_{i-1}$ and let~$j \in [i-1]$ be
  such that~$X_i \cap X \subseteq X_j$.  By induction hypothesis, we
  know that~$T_{i-1}$ is a~$K$-relation over~$X$ that witnesses the
  consistency of~$R_1,\ldots,R_{i-1}$.  First, we show that~$T_{i-1}$
  and~$R_i$ are consistent. By Lemma~\ref{lem:characterization} it
  suffices to show that~$T_{i-1}[X \cap X_i] \equiv R_i[X \cap
  X_i]$. Let~$Z = X \cap X_i$, so~$Z \subseteq X_j$ by the choice
  of~$j$, and indeed~$Z = X_j \cap X_i$.  Since~$j \leq i-1$, we
  have~$R_j \equiv T_{i-1}[X_j]$. By Part~3 of
  Lemma~\ref{lem:easyfacts2} and Part~3 of Lemma~\ref{lem:easyfacts1},
  we have~$R_j[Z] \equiv T_{i-1}[X_j][Z] = T_{i-1}[Z]$. By assumption,
  also~$R_j$ and~$R_i$ are consistent, and~$Z = X_j \cap X_i$, which
  by Lemma~\ref{lem:characterization} implies~$R_j[Z] \equiv
  R_i[Z]$. By transitivity, we get~$T_{i-1}[Z] \equiv R_i[Z]$, hence,
  by~$Z = X \cap X_i$ and Lemma~\ref{lem:characterization},
  the~$K$-relations~$T_{i-1}$ and~$R_i$ are consistent. We show
  that~$T_i = T_{i-1} \Join R_i$ witnesses the consistency
  of~$R_1,\ldots,R_i$. Since~$T_{i-1}$ and~$R_i$ are consistent, first
  note that~$T_{i-1} \equiv T_i[X]$ and~$R_i \equiv T_i[X_i]$ by
  Lemma~\ref{lem:characterization}. Now fix~$k \leq i-1$ and note that
  \begin{equation}
  R_k \equiv T_{i-1}[X_k] \equiv T_i[X][X_k] = T_i[X_k],
\end{equation}
where the first equivalence follows from the fact that~$T_{i-1}$ witnesses the
consistency of~$R_1,\ldots,R_{i-1}$ and~$k \leq i-1$, the second equivalence
follows from~$T_{i-1} \equiv T_i[X]$ together with Part~3 of
Lemma~\ref{lem:easyfacts2} applied to~$X_k \subseteq X$, and the equality
follows again from~$X_k \subseteq X$ and this time from Part~3 of
Lemma~\ref{lem:easyfacts1}. Thus, $T_i$ witnesses the consistency
of~$R_1,\ldots,R_i$, which was to be shown.
\end{proof}

In what follows, we explore the interplay between pairwise consistency
and global consistency of $K$-relations, aiming to extend Theorem
\ref{thm:BFMY} to arbitrary positive semirings.

\paragraph{Local-to-Global Consistency Property for~$K$-relations}
We extend the notion of local-to-global consistency property from
ordinary relations to $K$-relations. Let~$H$ be a hypergraph and
let~$X_1,\dots,X_m$ be a listing of all hyperedges of~$H$. We say
that~$H$ has the \emph{\ltgc~for~$K$-relations} if every pairwise
consistent collection~$R_1(X_1),\ldots,R_m(X_m)$ of~$K$-relations of
schema~$X_1,\ldots,X_m$ is globally consistent.

\begin{theorem} \label{thm:main}
  Let $K$ be a positive semiring and let $H$ be a hypergraph. The
  following statements are equivalent:
\begin{enumerate} \itemsep=0pt
\item [(a)] $H$ is an acyclic hypergraph.
\item [(b)] $H$ has the \ltgc~for $K$-relations.
\end{enumerate}
\end{theorem}

The claim that (a) implies (b) in Theorem~\ref{thm:main} follows
directly from the preceding Lemma~\ref{lem:newlemma}. We concentrate
on proving that (b) implies (a). To do this, we need four technical
lemmas. By Theorem \ref{thm:BFMY}, a hypergraph~$H$ is acyclic if and
only if it is conformal and chordal. The first two technical lemmas
state, in effect, that the ``minimal" non-conformal hypergraphs, as
well as the ``minimal" non-chordal hypergraphs, have very simple
forms.

\begin{lemma}[\cite{DBLP:journals/csur/Brault-Baron16}] \label{lem:characconf}
  A hypergraph~$H = (V,E)$ is not conformal if and only if there
  exists a subset $W$ of $V$ such that~$|W| \geq 3$
  and~$R(H[W])=(W,\{W\setminus\{A\} : A \in W\})$, where $R(H[W])$ is the reduction of the hypergraph $H[W]$ induced by $W$.
\end{lemma}

\begin{proof}
  The \emph{if} direction is immediate since, given that~$|W| \geq 3$,
  the set~$W$ forms a clique in the primal graph that is not included
  in any hyperedge of~$H$; otherwise no~$W\setminus\{A\}$
  with~$A \in W$ would be a hyperedge in the reduced hypergraph
  of~$H[W]$. For the \emph{only if} direction, let~$W$ be a clique in
  the primal graph of~$H$ that is not included in any hyperedge of~$H$
  and that is minimal with this property. Since the two vertices of
  every edge of the primal graph are included in some hyperedge of~$H$
  we have~$|W| \geq 3$. In addition, by minimality of~$W$,
  each~$W\setminus\{A\}$ with~$A \in W$ is included in some
  hyperedge~$X$ of~$A$ that does not contain~$W$,
  so~$X \cap W = W \setminus \{A\}$. This means that
  each~$W \setminus \{A\}$ is a hyperedge of~$H[W]$, and also of its
  reduced hypergraph since~$W$ is not included in any hyperedge
  of~$H$. Conversely, if~$X$ is a hyperedge in~$E$, then there is
  some~$A\in W$ such that~$X\cap W\subseteq W\setminus \{A\}$.  It
  follows that~$R(H[W])=(W,\{W\setminus\{A\} : A \in W\})$.
\end{proof}

\begin{lemma} \label{lem:characchord} A hypergraph $H = (V,E)$ is not
  chordal if and only if there exists a subset $W$ of $V$
  such that $|W| \geq 4$ and $R(H[W]) = (W,\{\{A_i,A_{i+1}\} : i \in [n]\})$, where~$A_1,\ldots,A_n$
  is an enumeration of $W$ and $A_{n+1} := A_1$.
\end{lemma}

\begin{proof}
  The \emph{if} direction is immediate since it implies that the
  primal graph of~$H$ contains a chordless cycle of length at least
  four. For the \emph{only if} direction, let~$W$ be the set of
  vertices of a shortest chordless cycle of length at least four in
  the primal graph of~$H$.
\end{proof}

The next two technical lemmas state that the local-to-global
consistency property for $K$-relations is preserved under induced
hypergraphs, and also under reductions.

\begin{lemma} \label{lem:preserv1}
  If a hypergraph~$H$ has
  the \ltgc~for~$K$-relations, then for every subset~$W$ of the
  vertices of~$H$ the hypergraph~$H[W]$ also has the
  \ltgc~for~$K$-relations.
\end{lemma}

\begin{proof}
  Assume that the hypergraph~$H=(V,E)$ has the
  \ltgc~for~$K$-relations. We will show that, for every
  vertex~$A\in V$, the
  hypergraph~$H[V\setminus \{A\}]= (V,\{X\setminus \{A\}: X \in E\})$
  also has the \ltgc~for~$K$-relations.  The statement of the lemma
  will follow from iterating this statement over all attributes~$A$
  in~$V \setminus W$.

  Let~$X_1,\ldots,X_m$ be a listing of all hyperedges of~$H$.  Fix a
  vertex~$A$ in~$V$ and write~$Y_i := X_i \setminus \{A\}$ for
  all~$i \in [m]$. Let~$R_1,\ldots,R_m$ be a collection of pairwise
  consistent~$K$-relations for~$Y_1,\ldots,Y_m$. Fix an arbitrary
  value~$u_0$ in the domain~$\domain(A)$ of the attribute~$A$. We
  define a collection of~$K$-relations~$S_1,\ldots,S_m$
  for~$X_1,\ldots,X_m$ as follows. For each~$i \in [m]$
  with~$A \not\in X_i$, let~$S_i := R_i$. For each~$i \in [m]$
  with~$A \in X_i$, let~$S_i$ be the~$K$-relation over~$X_i$ defined
  for every~$X_i$-tuple~$t$ by~$S_i(t) := 0$ if~$t(A) \not= u_0$ and
  by~$S_i(t) := R_i(t[Y_i])$ if~$t(A) = u_0$. We claim that
  the~$K$-relations~$S_1,\ldots,S_m$ are pairwise consistent.

     In order to see this, fix~$i,j \in [m]$ and distinguish the two cases
      whether~$A \not\in X_i X_j$ or~$A \in X_i X_j$: If~$A
     \not\in X_i X_j$, then~$S_i = R_i$ and~$S_j = R_j$ and
     therefore~$S_i$ and~$S_j$ are consistent because~$R_i$ and~$R_j$
     are consistent. If~$A \in X_i X_j$, then let~$R$ be
     a~$K$-relation over~$Y_i Y_j$ that witnesses the
     consistency~$R_i$ and~$R_j$ and let~$S$ be the~$K$-relation
     over~$X_i X_j$ defined for every~$X_i X_j$-tuple~$t$
     by~$S(t) := 0$ if~$t(A) \not= u_0$ and by~$S(t) := R(t[Y_i Y_j])$
     if~$t(A) = u_0$. We claim that~$S$ witnesses the
     consistency of~$S_i$ and~$S_j$. We show that~$S_i \equiv S[X_i]$
     and a symmetric argument will show that~$S_j \equiv S[X_j]$.  In
     order to see this, first we argue that~$R[Y_i] = S[Y_i]$.  Indeed,
     for every~$Y_i$-tuple~$r$ we have
     \begin{align}
      R(r) = \sum_{s \in R' : \atop s[Y_i] = r} R(s) =
      \sum_{t \in \tuples(X_i X_j) : \atop t[Y_i] = r, t(A) = u_0} R(t[Y_i Y_j]) =
      \sum_{t \in S' : \atop t[Y_i] = r} S(t) =
      S(r), \label{eqn:lsllsggg}
    \end{align}
    where the first equality follows from~\eqref{eqn:marginal}, the second
    follows from the fact that the map~$t \mapsto t[Y_i Y_j]$ is
    a bijection between the set of~$X_i X_j$-tuples~$t$ such
    that~$t[Y_i]=r$ and~$t(A) = u_0$ and the set
    of~$Y_i Y_j$-tuples~$s$ such that~$s[Y_i]=r$, the third
    follows from the definition of~$S$, and the fourth follows
    from~\eqref{eqn:marginal}. For later use, let us note that we did
    not assume that~$i \not= j$ for showing~\eqref{eqn:lsllsggg}. In
    case~$i = j$, the~$K$-relation~$R_i$ can serve as~$R$, and~$S$
    equals~$S_i$, which shows that~$R_i = S_i[Y_i]$.

    In case~$A \not\in X_i$, we have that~$Y_i = X_i$,
    hence Equation~\eqref{eqn:lsllsggg} already shows
    that~$S_i = R_i \equiv R[Y_i] = S[Y_i]$, so~$S_i \equiv
    S[X_i]$. In case~$A \in X_i$, let~$a,b \in K\setminus\{0\}$ be
    such that~$a R_i = b R[Y_i]$ and we show~$a S_i = b S[X_i]$. For
    every~$X_i$-tuple~$r$ with~$r(A) \not= u_0$, we have~$S_i(r) = 0$
    and also~$S(r) = \sum_{t : t[X_i] = r} S(t) = 0$
    since~$t[X_i] = r$ and~$A \in X_i$
    implies~$t(A) = r(A) \not= u_0$. Thus, $a S_i(r) = 0 = b S(r)$ in this
    case. For every~$X_i$-tuple~$r$ with~$r(A) = u_0$, we have
    \begin{align}
      a S_i(r) = a R_i(r[Y_i]) = b R(r[Y_i]) = b S(r[Y_i]), \label{lem:inter}
    \end{align}
    where the first equality follows from the definition of~$S_i$ and the
    assumption that~$r(A) = u_0$, the second follows from the choices
    of~$a$ and~$b$, and the third follows from~\eqref{eqn:lsllsggg}.
    Continuing from the right-hand side of~\eqref{lem:inter}, we have
    \begin{align}
    b S(r[Y_i]) = b \sum_{t \in S' : \atop t[Y_i] = r[Y_i]} S(t) =
    b \sum_{t \in S' : \atop t[X_i] = r} S(t) = b S(r),
    \label{lem:retni}
    \end{align}
    where the first equality follows from~\eqref{eqn:marginal}, the second
    follows from the assumption that~$A \in X_i$ and~$r(A) = u_0$
    together with~$S(t) = 0$ in case~$t(A) \not= u_0$, and the third
    follows from~\eqref{eqn:marginal}.  Combining~\eqref{lem:inter}
    with~\eqref{lem:retni}, we get~$a S_i(r) = b S(r)$ also in this
    case. This proves that~$S_i \equiv S[X_i]$;  a completely
    symmetric argument proves that~$S_j \equiv S[X_j]$.

    Since~$S_1,\ldots,S_m$ are pairwise consistent~$K$-relations
    for~$X_1,\ldots,X_m$, by assumption they are globally
    consistent. Let~$N$ be a~$K$-relation
    over~$X_1 \cup \cdots \cup X_m$ that witnesses their
    consistency. Let~$M := N[Y_1 \cup \cdots \cup Y_m]$ and we argue
    that~$M$ witnesses the consistency of~$R_1,\ldots,R_m$, which will
    prove the lemma. Fix~$i \in [m]$ and let~$a,b \in K\setminus\{0\}$
    be such that~$a S_i = b S[X_i]$. For every~$Y_i$-tuple~$r$ we have
   \begin{align}
     a R_i(r) = a S_i(r) = b N(r) = b M(r),
   \end{align}
   where the first follows~$S_i = R_i$ in case~$A \not\in X_i$ and
   from~\eqref{eqn:lsllsggg} applied to~$i = j$ and~$R = R_i$ in
   case~$A \in X_i$, the second follows from the choice of~$a$
   and~$b$, and the third follows from the choice of~$M$ and Part~3 of
   Lemma~\ref{lem:easyfacts1}.
\end{proof}

\begin{lemma} \label{lem:preserv2} If a hypergraph~$H$ has the
  \ltgc~for~$K$-relations, then the hypergraph~$R(H)$ also has
  the \ltgc~for~$K$-relations.
\end{lemma}

\begin{proof}
  Assume that the hypergraph~$(V,\{X_1,\ldots,X_m\})$ has the
  local-to-global consistency property for~$K$-relations. We will show
  that if~$X_m$ is covered by some other hyperedge,
  i.e.,~$X_m \subseteq X_i$ for some~$i \leq m-1$, then the
  hypergraph~$(V,\{X_1,\ldots,X_{m-1}\})$ also has the \ltgc~for~$K$
  relations. The statement of the lemma will follow from iterating
  this statement over all hyperedges of~$H$ that are covered by some
  other hyperedge of~$H$.

  Let~$R_1,\ldots,R_{m-1}$ be a collection of pairwise
  consistent~$K$-relations for~$X_1,\ldots,X_{m-1}$.
  Define~$R_m := R_i[X_m]$. We claim that~$R_1,\ldots,R_m$ are
  pairwise consistent. It suffices to check that~$R_j$ and~$R_m$ are
  consistent for any~$j \in [m]$ with~$j \not= m$.  By assumption, we
  know that~$R_j$ and~$R_i$ are consistent, which means that there
  exists a~$K$-relation~$T$ that witnesses their consistency; we
  have~$R_j \equiv T[X_j]$ and~$R_i \equiv T[X_i]$.
  Let~$S := T[X_j X_m]$. We have
  \begin{equation}
  R_j \equiv T[X_j] = T[X_j X_m][X_j] = S[X_j],
  \end{equation}
  where the first follows from the choice of $T$, the
  second follows from Part~3 of Lemma~\ref{lem:easyfacts1} and the
  third follows from the choice of~$S$.
  Likewise,
  \begin{equation}
  R_m = R_i[X_m] \equiv T[X_i][X_m] = T[X_m] = T[X_j X_m][X_m] = S[X_m],
  \end{equation}
  where the first equality is by  the choice of~$R_m$, the second follows from the
  choice of $T$, the assumption that~$X_m \subseteq X_i$, and Part~3
  of Lemma~\ref{lem:easyfacts2}, the third follows from the assumption
  that~$X_m \subseteq X_i$ and Part~3 of Lemma~\ref{lem:easyfacts1},
  the fourth follows again from Part~3 of Lemma~\ref{lem:easyfacts1},
  and the fifth follows from the choice of~$S$. Thus,~$S$ witnesses
  the consistency of~$R_j$ and~$R_m$.

  Since~$R_1\ldots,R_m$ are pairwise consistent~$K$-relations
  for~$X_1,\ldots,X_m$, by assumption they are globally
  consistent. The same~$K$-relation that witnesses their global
  consistency also witnesses the global consistency
  of~$R_1,\ldots,R_{m-1}$, which completes the proof.
\end{proof}

%
%

\paragraph{Generalized Tseitin Construction and Proof of Theorem~\ref{thm:main}}
The minimal non-conformal and minimal non-chordal hypergraphs from
Lemmas~\ref{lem:characconf}
and
Lemma \ref{lem:characchord}
share the following
properties: 1)~all their hyperedges have the same number of vertices,
and 2)~all their vertices appear in the same number of hyperedges. For
hypergraphs~$H$ that have these properties, we construct
a collection~$C(H; K)$ of~$K$-relations that are indexed by the
hyperedges of~$H$; these relations will play a crucial role in the proof of
Theorem~\ref{thm:main}.

Let~$H = (V,E)$ be a hypergraph and let~$d$ and~$k$ be positive
integers. The hypergraph~$H$ is called~\emph{$k$-uniform} if every
hyperedge of~$H$ has exactly~$k$ vertices. It is called
\emph{$d$-regular} if any vertex of~$H$ appears in exactly~$d$
hyperedges of~$H$. For example, the non-conformal hypergraph of
Lemma~\ref{lem:characconf} is~$k$-uniform and~$d$-regular
for~$k := d := |W|-1$. Likewise, the non-chordal hypergraph of
Lemma~\ref{lem:characchord} is~$k$-uniform and~$d$-regular
for~$k := d := 2$.  For a positive semiring~$K$ and each~$k$-uniform
and~$d$-regular hypergraph~$H$ with~$d \geq 2$ and with
hyperedges~$E = \{X_1,\ldots,X_m\}$, we construct a collection
of~$K$-relations~$C(H; K) := \{ R_1(X_1),\ldots,R_m(X_m) \}$,
where~$R_i$ is a~$K$-relation that has attributes~$X_i$. The
collection~$C(H; K)$ of these~$K$-relations will turn out to be pairwise consistent but
not globally consistent. Note that by the characterization of
acyclicity in terms of Graham's algorithm, a hypergraph that
is~$k$-uniform and~$d$-regular for some~$k \geq 2$ and~$d \geq 2$
will never be acyclic: Graham's procedure will not even start to
remove any hyperedge or any vertex. Hence, the existence of
the~$K$-relations~$R_1,\ldots,R_m$ that violates the local-to-global
consistency property is compatible with Theorem~\ref{thm:main}. The
construction is defined as follows.

For each~$i \in [m]$ with~$i \not= m$, let~$R_i$ be the unique
unit~$K$-relation over~$X_i$ whose support contains all
tuples~$t : X_i \rightarrow \{0,\ldots,d-1\}$ whose total
sum~$\sum_{C \in X_i} t(C)$ is congruent to~$0$ mod~$d$;
i.e.,~$R(t) := 1$ for each such~$X_i$-tuple, and~$R(t) := 0$ for any
other~$X_i$-tuple.  For~$i = m$, let~$R_i$ be the unique
unit~$K$-relation over~$X_i$ whose support contains all
tuples~$t : X_i \rightarrow \{0,\ldots,d-1\}$ whose total
sum~$\sum_{C \in X_i} t(C)$ is congruent to~$1$ mod~$d$; i.e.,
again~$R(t) := 1$ for each such~$X_i$-tuple, and~$R(t) := 0$
otherwise.

To show the pairwise consistency of~$R_1,\ldots,R_m$, it suffices, by
Lemma~\ref{lem:characterization}, to show that for every two
distinct~$i,j \in [m]$, we have~$R_i[Z] \equiv R_j[Z]$,
where~$Z := X_i \cap X_j$. In turn, this follows from the claim that
for every~$Z$-tuple~$t : Z \rightarrow \{0,\ldots,d-1\}$, we
have~$R_i(t) = R_j(t) = N_Z 1 = 1 + \cdots + 1$, the sum
of~$N_Z := d^{k-|Z|-1}$ many units of the semiring~$K$. Indeed, since
by~$k$-uniformity every hyperedge of~$H$ has exactly~$k$ vertices, for
every~$u \in \{0,\ldots,d-1\}$, there are exactly~$N_Z$
many~$X_i$-tuples~$t_{i,u,1},\ldots,t_{i,u,N_Z}$ that extend~$t$ and
have total sum congruent to~$u$ mod~$d$. It follows then
that~$R_i[Z] = R_j[Z]$ regardless of whether~$n \in \{i,j\}$
or~$n \not\in \{i,j\}$, and hence any two~$R_i$ and~$R_j$ are
consistent by Lemma~\ref{lem:characterization}.
To argue that the relations~$R_1,\ldots,R_m$ are not globally
consistent, we proceed by contradiction. If~$R$ were a~$K$-relation
that witnesses their consistency, then it would be non-empty and its
support would contain a tuple~$t$ such that the projections~$t[X_i]$
belong to the supports~$R'_i$ of the~$R_i$, for each~$i \in [m]$. In
turn this means that
  \begin{align}
    & \textstyle{\sum_{C \in X_i} t(C) \;\equiv\; 0 \text{ mod } d}, \;\;\;\;\;\
 \text{ for $i \not = m$ } \label{eqn:restlu} \\
    & \textstyle{\sum_{C \in X_i} t(C) \;\equiv\; 1 \text{ mod } d}, \;\;\;\;\;\
 \text{ for $i = m$. } \label{eqn:a0lu}
  \end{align}
  Since by $d$-regularity each $C \in V$ belongs to exactly $d$ many
  sets $X_i$, adding up all the equations in~\eqref{eqn:restlu}
  and~\eqref{eqn:a0lu} gives
\begin{equation}
    \textstyle{\sum_{C \in V} dt(C) \;\equiv\; 1 \text{ mod } d},
  \end{equation}
  which is absurd since the left-hand side is congruent to~$0$
  mod~$d$, the right-hand side is congruent to~$1$ mod~$d$,
  and~$d \geq 2$ by assumption.

We now have all the tools needed to present the proof of Theorem \ref{thm:main}.

\begin{proof}[Proof of Theorem~\ref{thm:main}]
  As stated earlier, the direction (a) implies (b) follows from
  Lemma~\ref{lem:newlemma}.  For showing that (b) implies (a), let us
  assume the contrary and then we will derive a contradiction.
  Let~$H = (V,E)$ be a smallest counterexample to the statement
  that~(b) implies~(a), meaning that the following three conditions
  hold: (i)~$H$ is not acyclic, (ii)~$H$ has the
  local-to-global consistency property for~$K$-relations, and
  (iii)~$H$ is minimal in the sense that~$n = |V|$ is smallest
  possible with properties~(i) and~(ii), and among those,~$m = |E|$ is
  smallest possible.  Since~$H$ is not acyclic, we know that~$H$ is
  either not conformal or not chordal. We distinguish the two cases.
  \bigskip

 \noindent\textit{Case 1:}~$H$ is not conformal. By Lemmas~\ref{lem:preserv1}
 and~\ref{lem:preserv2}, the minimality of~$H$, and
 Lemma~\ref{lem:characconf}, we have~$n \geq 3$ and~$m = n$;
 indeed,~$E = \{ V_i : i \in [n] \}$ where~$V_i = V\setminus\{A_i\}$
 and~$A_1,\ldots,A_n$ is an enumeration of~$V$. Thus,~$H$
 is~$k$-uniform and~$d$-regular for~$k = d = n-1 \geq 2$. The
 construction~$C(H; K)$ gives a collection
 of~$K$-relations~$R_1,\ldots,R_n$, where~$R_i$ has attributes~$V_i$,
 which are pairwise consistent but not globally consistent, which is a
 contradiction.  \bigskip

 \noindent \textit{Case 2:} $H$ is not chordal. By
 Lemmas~\ref{lem:preserv1} and~\ref{lem:preserv2}, the minimality
 of~$H$, and Lemma~\ref{lem:characchord}, we have~$n \geq 4$
 and~$m = n$, and indeed~$E = \{ V_i : i \in [n] \}$
 where~$V_i = \{ A_i,A_{i+1} \}$ and~$A_1,\ldots,A_n$ is an
 enumeration of~$V$ with~$A_{n+1} := A_1$. Thus, $H$ is $k$-uniform
 and $d$-regular for $k = d = 2$. Again, the construction $C(H; K)$
 gives a collection of $K$-relations $R_1,\ldots,R_n$, where~$R_i$ has
 attributes~$V_i$, which are pairwise consistent but not globally
 consistent, which is a contradiction.
\end{proof}

An inspection of the proof of Theorem~\ref{thm:main} reveals that
actually a stronger result is established. Specifically, let~$H$ be a
hypergraph and let~$X_1,\ldots,X_m$ be its hyperedges. The proof of
Theorem~\ref{thm:main} shows that if~$H$ is not acyclic, then there
are ordinary relations~$R_1(X_1),\ldots,R_m(X_m)$ such that, for every
positive semiring~$K$, the
unit~$K$-relations~$S_1(X_1),\ldots,S_m(X_m)$ with~$S_i'=R_i$ are
pairwise consistent but globally inconsistent.
%
%
In more
informal terms, the proof of Theorem \ref{thm:main} actually shows
that if a hypergraph is acyclic, then there is an essentially
\emph{uniform} counterexample to the \ltgc~for~$K$-relations that
works for all positive semirings~$K$.

Before establishing the main result in this paper, we bring into the picture one more notion from hypergraph theory that was introduced by
Vorob'ev \cite{vorob1962consistent} in his study of global consistency for probability distributions.

\paragraph{Vorob'ev Regular Hypergraphs}
A \emph{complex} is a hypergraph~$H=(V,E)$ whose set~$E$ of hyperedges
is closed under taking subsets, i.e., if~$X\in E$ and~$Z\subseteq X$,
then~$Z\in E$. The \emph{downward closure} of a hypergraph~$H$ is the
hypergraph whose vertices are those of~$H$ and whose hyperedges are
all the subsets of the hyperedges of~$H$. Clearly, the downward
closure of a hypergraph is a complex.

Let~$K$ be a complex. Let~$X$ and~$Y$ be two different maximal hyperedges
of~$K$, where a hyperedge is \emph{maximal} if it not a proper subset of any other hyperedge.  We say that~$X$ yields a maximal intersection with~$Y$ if the
intersection~$X \cap Y$ is not a proper subset of the intersection~$X
\cap Z$ of some third
\footnote{The condition ``third'' is missing in
  Vorob'ev \cite{vorob1962consistent}, and it is necessary since~$e \cap e'$ is always a proper
  subset of~$e \cap e''$ whenever~$e = e''$ and~$e$ is maximal.}
  hyperedge~$Z$ of~$K$. A maximal hyperedge~$X$ of~$K$ is called \emph{extreme
  in}~$K$ if all maximal intersections of~$X$ with hyperedges of~$K$ are
  equal. Let~$X$ be an extreme hyperedge in~$K$. The \emph{proper} vertices of~$X$
  are those that do not belong to any other maximal hyperedge of~$K$. The
  \emph{normal subcomplex} of~$K$ corresponding to the extreme edge~$X$ is
  the subcomplex of~$K$ consisting of all hyperedges of~$K$ that do not
  intersect the set of proper vertices of~$e$. A subcomplex~$K'$
  of~$K$ is called a \emph{normal subcomplex} if it is the normal subcomplex
  corresponding to some extreme hyperedge of~$K$. A \emph{normal series} of~$K$ is
  a sequence of subcomplexes
\begin{equation}
K = K_0 \supset K_1 \supset \cdots \supset K_r \label{eqn:series}
\end{equation}
of the complex~$K$ in which for every $\ell$ with $1\leq \ell \leq r-1$, the
complex~$K_{\ell+1}$ is a normal subcomplex of the complex~$K_\ell$,
and the final complex~$K_r$ does not have any extreme hyperedges. We say
that~$K$ is \emph{regular} if there exists a normal series of~$K$ in
which the last term is the complex without vertices.  We say that a
hypergraph~$H$ is \emph{Vorob'ev regular} if its downward closure is a
regular complex.

We will show that a hypergraph is Vorob'ev regular if and only if it is acyclic. This result will follow from the next two lemmas and Theorem \ref{thm:BFMY}. We have not been able to locate a published proof of this result in the literature, even though the equivalence between acyclicity and Vorob'ev regularity  has been mentioned  in \cite{hill1991relational,DBLP:journals/sigact/Yannakakis96}.

\begin{lemma} \label{lem:vor-grah}
If~$H$ is a Vorob'ev regular hypergraph, then Graham's
algorithm succeeds on~$H$.
\end{lemma}

\begin{proof}
Let~$H$ be a Vorob'ev regular hypergraph. The proof that
Graham's algorithm succeeds on~$H$ is by induction on the length of a
normal series of the downward closure~$K$ of~$H$.  If the length of a
normal series of~$K$ is zero, then~$K$ and hence~$H$ itself is the
empty hypergraph and there is nothing to prove.  Assume then that~$K$
has a normal series that has length at least one, let~$K_1$ be the
first subcomplex of~$K$ in the series, and let~$X$ be the extreme hyperedge
of~$K$ corresponding to which~$K_1$ is its normal
subcomplex. Then~$K_1$ consists of the hyperedges of~$K$ that do not
intersect the proper vertices of~$e$. Equivalently,~$K_1$ is obtained
from~$K$ by deleting all the proper vertices of~$e$. Since the proper
vertices of~$e$ appear in no other maximal hyperedge of~$K$, this means
that if we delete the proper vertices of~$e$ from all the hyperedges of~$H$
in which they appear, then we obtain a hypergraph whose downward
closure is~$K_1$. Moreover, such a hypergraph~$H_1$ can be obtained
from~$H$ by a applying a sequence of operations of Graham's algorithm:
first delete all the hyperedges that are proper subsets of~$e$, then delete
all the proper vertices of~$X$. Now,~$K_1$ is also Vorob'ev's regular,
and its normal series has length one less than that of~$K$. Hence, by
induction hypothesis, Graham's algorithm succeeds on~$H_1$, which
means that there is a sequence of operations of Graham's algorithm
that applied to~$H_1$ yield the empty hypergraph. By concatenating the
two sequences of operations, we get a single sequence of operations of
Graham's algorithm that starts at~$H$ and yields the empty
hypergraph. This proves that Graham's algorithm succeeds on~$H$.
\end{proof}

\begin{lemma} \label{lem:join-vor}
If~$H$ is a hypergraph that has a join tree, then~$H$ is
Vorob'ev regular.
\end{lemma}

\begin{proof}
  Let~$H$ be a hypergraph that has a join tree.  By induction on the
  number of its hyperedges, we show that~$H$ is Vorob'ev
  regular. Let~$K$ be the downward closure of~$H$. If~$H$ has no
  hyperedges, then it is Vorob'ev regular. If~$H$ has just one
  hyperedge~$X$, then~$X$ is an extreme hyperedge of~$K$ since it
  yields no maximal intersections at all. It follows that~$K$ is
  regular since all the vertices of~$e$ are proper, so~$H$ is Vorob'ev
  regular. Assume now that~$H$ has more than one hyperedges. Let~$X$
  be a leaf of the join tree of~$H$, and let~$Y$ be the unique
  hyperedge of~$H$ such that~$\{X,Y\}$ is an edge of the join tree. We
  consider the following two cases.
  \bigskip

  \noindent \textit{Case 1:} If~$X$ is not a maximal hyperedge,
  say~$X \subseteq Z$ for some other edge~$Z$ of~$H$,
  then~$X = X \cap Z \subseteq X \cap Y$ by the connectivity property
  of the join tree since the unique path from~$X$ to~$Z$ in the tree
  must pass through~$Y$. Hence, $X \subseteq Y$. Now, let~$H'$ be the
  hypergraph that results from deleting~$X$ from~$H$.
  Since~$X \subseteq Y$, the tree that results from trimming the
  leaf~$X$ from the join tree of~$H$ is a join tree of~$H'$.  By
  induction hypothesis,~$H'$ is Vorob'ev regular. But~$X$ was not
  maximal in~$H$, so~$H$ and~$H'$ have the same downward closure~$K$,
  which shows that~$H$ is also Vorob'ev regular.
  \bigskip

  \noindent \textit{Case 2:} If~$X$ is a maximal hyperedge of~$H$,
  then we claim that~$X$ is an extreme hyperedge of~$K$. First,~$X$ is
  a maximal hyperedge of~$K$ by assumption.  Second, we show that~$X$
  yields maximal intersection with~$Y$. Indeed, if~$Z$ is a third
  hyperedge of~$H$, then the unique path from~$X$ to~$Z$ in the join
  tree goes through~$Y$ and, by the connectivity property of the join
  tree, every vertex in~$X \cap Z$ belongs to~$Z \cap Y$.
  So~$X \cap Y$ is not a proper subset of~$X \cap Z$. Next, we show
  that all maximal intersections of~$X$ are equal to the maximal
  intersection of~$X$ with~$Y$. Let~$Z$ be a third hyperedge of~$H$
  and assume that~$X$ yields a maximal intersection with~$Z$. In
  particular, $X \cap Z$ is a not a proper subset of~$X \cap Y$. But
  the connectivity property of the join tree
  implies~$X \cap Z \subseteq X \cap Y$ since the unique path from~$X$
  to~$Z$ in the join tree goes through~$Y$.
  Thus~$X \cap Z = X \cap Y$.

  We proved that~$X$ is an extreme hyperedge of~$K$. Now, let~$H'$ be
  the hypergraph that is obtained from~$H$ by deleting~$X$ and all the
  vertices that appeared only in~$X$. The normal subcomplex~$K'$
  of~$K$ corresponding to~$X$ is the downward closure of~$H'$. If we
  trim the leaf~$X$ from the join tree of~$H$, we get a join tree
  of~$H'$ with one node less. It follows from the induction hypothesis
  that~$H'$ is Vorob'ev regular. Thus~$K'$ is Vorob'ev regular, and
  so is~$K$ since~$K'$ is a normal subcomplex of it.
\end{proof}

Lemma \ref{lem:vor-grah}, Lemma \ref{lem:join-vor}, and Theorem \ref{thm:BFMY} imply the following result.

\begin{corollary} \label{cor:vorobev}
Let $H$ be a hypergraph. The following statements are equivalent:
\begin{enumerate} \itemsep=0pt
\item[(a)] $H$ is an acyclic hypergraph.
\item[(b)] $H$ is a Vorob'ev regular hypergraph.
\end{enumerate}
\end{corollary}

Finally, by combining Theorem \ref{thm:BFMY}, Theorem \ref{thm:main},
and Corollary \ref{cor:vorobev}, we obtain the main result of this
paper.

\begin{theorem} \label{thm:BFMY-ext} Let $K$ be a positive semiring
  and let $H$ be a hypergraph. The following statements are
  equivalent:
\begin{enumerate} \itemsep=0pt
\item[(a)] $H$ is an acyclic hypergraph.
\item[(b)] $H$ is a conformal and chordal hypergraph.
\item[(c)] $H$ has the running intersection property.
\item[(d)] $H$ has a join tree.
\item[(e)] $H$ is accepted by Graham's algorithm.
\item[(f)] $H$ is a Vorob'ev regular hypergraph.
\item[(g)] $H$ has the \ltgc~for $K$-relations.
\end{enumerate}
\end{theorem}

By applying Theorem~\ref{thm:BFMY-ext} with~$K$ set to the Boolean
semiring~$\mathbb{B}$, we obtain the original
Beeri-Fagin-Maier-Yannakakis  Theorem~\ref{thm:BFMY}. It should
be noted that while our proof of Theorem~\ref{thm:BFMY-ext} actually
used Theorem~\ref{thm:BFMY} several times, all the uses
of~Theorem~\ref{thm:BFMY} made were for arguing that the various syntactic
characterizations of hypergraph acyclicity are equivalent. These
equivalences can be shown directly without referring to any semantic
notion, and for this reason we can say that our proof does
\emph{not} rely on the semantic part of Theorem~\ref{thm:BFMY}. In fact, the
main construction that we gave in the proof of Theorem~\ref{thm:main}
appears to be new and gives an alternative proof of the semantic part
of Theorem~\ref{thm:BFMY}: If~$H$ has \ltgc~for ordinary relations, then~$H$
is acyclic.

While still different, the construction we gave in the proof of
Theorem~\ref{thm:main} is closer to Vorob'ev's proof of his theorem
from~\cite{vorob1962consistent} than to the proof of the  Theorem~\ref{thm:BFMY}
from~\cite{BeeriFaginMaierYannakakis1983}. We discuss this next.

\subsection{Consistency of Probability Distributions and Vorob'ev's Theorem}

In the rest of this section, we study the consistency of~$K$-relations
when the semiring~$K$ is~$\mathbb{R}^{\geq 0}$, i.e., the set of
non-negative real numbers with the standard addition and multiplication
operations. We place the focus on probability distributions, which, to
recall, are nothing but the~$\mathbb{R}^{\geq 0}$-relations~$T$ that
satisfy the normalization constraint
\begin{equation}
T[\emptyset] = \sum_{t \in T'} T(t) = 1. \label{eqn:normalization}
\end{equation}
Our goal for the rest of this section is to show that the main
result of Vorob'ev from \cite{vorob1962consistent} follows from
our general result Theorem~\ref{thm:BFMY-ext} about arbitrary
positive semirings.

\paragraph{Consistency of Probability Distributions}
Following~\cite{vorob1962consistent}, we say that two probability
distributions~$P(X)$ and~$Q(Y)$ are consistent if there exists a
probability distribution~$T(XY)$ such that~$T[X]=P$ and~$T[Y]=Q$. A
collection~$P_1(X_1),\ldots,P_m(X_m)$ of probability distributions
is~\emph{globally consistent} if there exists a probability
distribution~$P(X_1 \cdots X_m)$ such that~$P[X_i] = P_i$ holds for
every~$i \in [m]$. The collection is called~\emph{pairwise consistent}
if any two distributions in the collection are consistent. We start by
showing that, when the probability distributions are presented
as~$\mathbb{R}^{\geq 0}$-relations that
satisfy~\eqref{eqn:normalization}, this classical notion of
consistency of probability distributions coincides with the notion of
consistency that we have been studying in this paper. The following
basic fact was observed already
in~Section~\ref{sec:valuedrelations}. It says that,
as~$\mathbb{R}^{\geq 0}$-relations, the probability distributions are
the canonical representatives of their equivalence classes
under~$\equiv$. We give an even shorter proof in a slightly different
language.

\begin{lemma} \label{lem:loosevsstrict} For every two probability
  distributions~$R(X)$ and~$S(X)$ over the same set of attributes~$X$,
  it holds that~$R \equiv S$ if and only if~$R = S$.
\end{lemma}

\begin{proof}
  The \emph{if} direction is trivial. For the \emph{only if}
  direction, let~$a$ and~$b$ be positive reals such
  that~$aR(t) = bS(t)$ holds for every~$X$-tuple $t$.
  Then~$a = aR[\emptyset] = bS[\emptyset] = b$, where the first
  follows from~\eqref{eqn:normalization}, the second follows
  from~\eqref{eqn:marginal} and the choice of~$a$ and~$b$, and the
  third follows from~\eqref{eqn:normalization}. Dividing through
  by~$a = b \not= 0$, it follows that~$R(t) = S(t)$ holds for
  every~$X$-tuple~$t$.
\end{proof}

The next lemma states that any~$\mathbb{R}^{\geq 0}$-relation that
witnesses the consistency of a collection of probability distributions
is itself a probability distribution.

\begin{lemma} \label{lem:closure} For every collection~$R_1,\ldots,R_m$
  of probability distributions and
  every~$\mathbb{R}^{\geq 0}$-relation~$R$, if~$R$ witnesses the global
  consistency of~$R_1,\ldots,R_m$, then~$R$ is itself a probability
  distribution.
\end{lemma}

\begin{proof}
  For any~$i \in [m]$ we have~$R[X_i] = R_i$, where~$X_i$ is the set
  of attributes of~$R_i$. By Part~3 of Lemma~\ref{lem:easyfacts1} we
  have~$R[\emptyset] = R[X_i][\emptyset] = R_i[\emptyset] = 1$, for
  any~$i \in [m]$; i.e.,~$R$ satisfies~\eqref{eqn:normalization} and
  is hence a probability distribution.
\end{proof}

In view of Lemma~\ref{lem:loosevsstrict} and~\ref{lem:closure}, the
classical notions of consistency of probability distributions
coincides with the notions of consistency of~$K$-relations that we
have been studying in this paper when~$K = \mathbb{R}^{\geq 0}$. We
are ready to state the \ltgc~for probability distributions.

\paragraph{Vorob'ev's Theorem}
Let~$H$ be a hypergraph and let~$X_1,\ldots,X_m$ be a listing of all
its hyperedges. We say that~$H$ has the \emph{\ltgc~for probability
  distributions} if every pairwise consistent collection of
probability distributions~$P_1(X_1),\ldots,P_m(X_m)$ is globally
consistent.

One of the main motivations for writing this paper was to obtain a
common generalization of the result of Beeri, Fagin, Maier, and
Yannakakis stated in Theorem~\ref{thm:BFMY} and the result of~Vorob'ev
stated next.

\begin{theorem}[Theorem 4.2 in
  \cite{vorob1962consistent}] \label{thm:vorothm} Let $H$ be a
  hypergraph. The following statements are equivalent:
\begin{enumerate} \itemsep=0pt
\item[(a)] $H$ is a Vorob'ev regular hypergraph.
\item[(b)] $H$ has the local-to-global consistency property for probability distributions.
\end{enumerate}
\end{theorem}

\begin{proof} By~Lemmas~\ref{lem:loosevsstrict} and~\ref{lem:closure},
  conditions~(b) in this Theorem and~(g) in Theorem~\ref{thm:BFMY-ext}
  for~$K = \mathbb{R}^{\geq 0}$ are equivalent. The result now follows
  from Theorem~\ref{thm:BFMY-ext} for $K = \mathbb{R}^{\geq 0}$.
\end{proof}

Some words on the differences between our proof of
Theorem~\ref{thm:vorothm} and Vorob'ev's proof of Theorem~4.2 in
\cite{vorob1962consistent} are in order. In the direction~(a)
implies~(b), except for the minor differences that stem from the use
of Lemma~\ref{lem:vor-grah}, our proof is basically the same as
Vorob'ev's. The main construct in that proof is the operation on
probability distributions that we denoted~$\Join_{\mathrm{P}}$ in
equation~\eqref{eqn:prop1}, which appears with different notation as
equation~(21) in page 156 of~\cite{vorob1962consistent}.

In the direction~(b) implies~(a), again except for the minor
differences that stem from the use of Lemma~\ref{lem:join-vor}, our
proof has some important similarities with Vorob'ev's, but also one
important difference. The similarities lie in the structure of the
argument. Vorob'ev first proves that the class of hypergraphs that
have the \ltgc~for probability distributions is the \emph{unique}
class of hypergraphs that contains those, satisfies certain closure
properties, and excludes two concrete families of hypergraphs, that he
calls~$\{G_n\}$ and~$\{Z_n\}$ in Theorem~2.2 in page~152
of~\cite{vorob1962consistent}. This characterization we also prove
through the combination of Lemmas~\ref{lem:preserv1}
and~\ref{lem:preserv2}, which stand for the closure properties,
Lemmas~\ref{lem:characconf} and~\ref{lem:characchord}, whose featuring
hypergraphs are precisely the hypergraphs~$G_n$ and~$Z_n$
from~\cite{vorob1962consistent}, and the construction~$C(H;K)$ through
Case~1 for~$H=G_n$, and Case~2 for~$H=Z_n$, in the proof of
Theorem~\ref{thm:main}.  Another similarity lies in the way we
handle~$Z_n$: both proofs build a cycle of~$K$-relations that
implement equality constraints except for one~$K$-relation in the
cycle that implements an inequality constraint. The important
difference lies in the way we handle~$G_n$. Vorob'ev's proof is a
linear-algebraic argument over Euclidean space that requires some
non-trivial calculations, while our argument is more combinatorial and
arguably simpler as it relies on basic modular arithmetic and the
totally obvious fact that~$0 \not\equiv 1$ mod~$d$, for~$d \geq 2$.

\section{Concluding Remarks}

We conclude  by discussing some open problems and directions for future research.

\begin{itemize}

\item Beeri et al.\ \cite{BeeriFaginMaierYannakakis1983} showed that
  hypergraph acyclicity is also equivalent to certain semantic
  conditions other than the \ltgc~for ordinary relations. The
  existence of a \emph{full reducer} is arguably the most well known
  and useful such semantic property (see also
  \cite{DBLP:books/cs/Maier83,DBLP:books/cs/Ullman88}). By definition,
  a \emph{full reducer} of a hypergraph~$H$ with~$X_1,\ldots,X_m$ as
  its hyperedges is a program consisting of a finite sequence of
  semijoin statements of the form~$R_i := R_i\ltimes R_j$ such that if
  this program is given a collection~$R_1(X_1),\ldots,R_m(X_m)$ of
  ordinary relations as input, then the output is a collection of
  globally consistent ordinary relations.

  It remains an open problem to define a suitable semijoin operation
  for~$K$-relations and prove (or disprove) that for every positive
  semiring~$K$, a hypergraph~$H$ is acyclic if and only if~$H$ has a
  full reducer for~$K$-relations. One of the technical difficulties is
  that the join operation on two~$K$-relations introduced and studied
  here is not, in general, associative (in fact, as seen earlier, it
  is not associative even when~$K$ is the bag semiring of non-negative
  integers).

\item The notion of consistency studied here is based on the notion of
  equivalence~$\equiv$ of two~$K$-relations. One could define the
  notion of \emph{strict consistency}, where two~$K$-relations~$R(X)$
  and~$S(Y)$ are \emph{strictly consistent} if there is a $K$-relation $T(XY)$ such that
  $T[X] = R$ and $T[Y]=S$;  from there, one can define the notion
  of \emph{local-to-global strict consistency property}
  for~$K$-relations. The question is: do the main results presented here hold for
  strict consistency? In particular, is hypergraph acyclicity
  equivalent to the local-to-global strict consistency property
  for~$K$-relations?

  The crucial obstacle in following the approach we adopted here
   is that if~$K$ is not a semifield, then it is not clear how
  to define a join operation on two~$K$-relations that witnesses their
  strict consistency. Actually,  there are  bags~$R$ and~$S$ that
  are strictly consistent, but every bag~$T$ that witnesses their
  strict consistency (as defined above) has the property that its support~$T'$  is \emph{strictly
    contained} in the ordinary join~$R' \Join S'$ of the supports~$R'$
  and~$S'$. This is in sharp contrast with our join operation, as shown in~Lemma~\ref{lem:supportsjoin}. The simplest such example is
  the pair of bags~$R(AB) = \{ {12}:{1},\; {22}:{1} \}$
  and~$S(BC) = \{ {21}:{1},\; {22}:{1} \}$, whose strict consistency
  is witnessed by the bags~$T_1(ABC) = \{ {122}:{1},\; {221}:{1} \}$
  and~$T_2(ABC) = \{ {121}:{1},\; {222}:{1} \}$, but, as one can easily verify,
  no other bag. 

\item Positive semirings have been used in different areas of
  mathematics and computer science (see, e.g.,
  \cite{golan2013semirings,gondran2008graphs}). In particular, the
  min-plus semiring has been used in the analysis of dynamic
  programming algorithms, while the Viterbi semiring has been used in
  the study of statistical models. It would be interesting to
  investigate potential applications of the results reported here to
  semirings other than the Boolean semiring and the semiring of the
  non-negative real numbers with the standard arithmetic operations.
    \end{itemize}

\paragraph{Acknowledgments} Albert Atserias  was partially funded by
European Research Council (ERC) under the European Union's Horizon
2020 research and innovation programme, grant agreement ERC-2014-CoG
648276 (AUTAR), and by MICCIN grant PID2019-109137GB-C22 (PROOFS).
Phokion Kolaitis was partially supported by NSF Grant  IIS-1814152.

\newpage
\bibliographystyle{alpha}
\bibliography{biblio}

\end{document}